\documentclass[a4paper,leqno]{amsart}
\usepackage{latexsym}
\usepackage[english]{babel}
\usepackage{fancyhdr}
\usepackage[mathscr]{eucal}
\usepackage{amsmath}
\usepackage{mathrsfs}
\usepackage{amsthm}
\usepackage{amsfonts}
\usepackage{amssymb}
\usepackage{amscd}
\usepackage{bbm}
\usepackage{graphicx}
\usepackage{graphics}
\usepackage{latexsym}
\usepackage{color}
 \newtheorem{thm}{Theorem}[section]
 \newtheorem{cor}[thm]{Corollary}
 \newtheorem{lem}[thm]{Lemma}
 \newtheorem{prop}[thm]{Proposition}
 \theoremstyle{definition}
 
 \theoremstyle{remark}
 \newtheorem{rem}[thm]{Remark}
 
 \numberwithin{equation}{section}

\newcommand{\ii}{\mathrm{i}}
\newcommand{\cH}{\mathcal{H}}

 \newcommand{\ud}{\mathrm{d}}
\def\a{\alpha}

\def\Ga{\Gamma}

\def\ep{\varepsilon}

\def\th{\theta}

\def\k{\kappa}

\def\r{\rho}

\def\w{\omega}
\def\W{\Omega}
\def\Hg {{\mathcal H}}
\def\la{\langle}
\def\ra{\rangle}
\def\supp{\mathop{\rm supp} \nolimits}
\def\lap{\Delta}
\def\ax{{\la x \ra}}
 
\begin{document}

%
%
%
%
%
%
%
%
%

\title[$L^p$-boundedness of wave operators for multi-centre]
 {$L^p$-boundedness of wave operators  \\ for the three-dimensional multi-centre \\ point interaction}

\author[G.~Dell'Antonio]{Gianfausto Dell'Antonio}

\address[G.~Dell'Antonio]{%
Department of Mathematics, University of Rome La Sapienza \\ Piazza Aldo Moro 5, 00185 Rome (ITALY), and \\ SISSA -- International School for Advanced Studies \\ via Bonomea 265, 34136 Trieste (ITALY)}

\email{gianfa@sissa.it}

\thanks{Partially supported by the 2014-2017 MIUR-FIR 
grant ``\emph{Cond-Math: Condensed Matter and Mathematical Physics}'' code RBFR13WAET and by \emph{JSPS grant in aid 
for scientific research} No.16K05242.}
\author[A.~Michelangeli]{Alessandro Michelangeli}

\address[A.~Michelangeli]{%
SISSA -- International School for Advanced Studies \\ via Bonomea 265, 34136 Trieste (ITALY)}

\email{alemiche@sissa.it}

\author[R.~Scandone]{Raffaele Scandone}

\address[R.~Scandone]{%
SISSA -- International School for Advanced Studies \\ via Bonomea 265, 34136 Trieste (ITALY)}

\email{rscandone@sissa.it}

\author[K.~Yajima]{Kenji Yajima}

\address[K.~Yajima]{%
Department of Mathematics, Gakushuin University \\ 1-5-1 Mejiro, Toshima-ku, Tokyo 171-8588 (JAPAN)}

\email{kenji.yajima@gakushuin.ac.jp}




\date{21 April 2017}

\begin{abstract}
We prove that, for arbitrary centres and strengths, 
the wave operators for three dimensional 
Schr\"odinger operators with multi-centre local point 
interactions are bounded in $L^p(\mathbb{R}^3)$ for $1<p<3$ and unbounded otherwise. 
\end{abstract}

\maketitle

\section{Introduction and main results}

Models of quantum particles in $d$ dimensions which scatter freely in space except for the presence of a number of extremely localised impurities require the definition of a Hamiltonian that acts precisely as the free Hamiltonian on wave-functions supported away from the scattering centres, and that induces a non-trivial interaction essentially supported on a discrete collection of points. This naturally leads to consider `singular' perturbations of the free Schr\"{o}dinger operator which can be thought of as delta-like potentials centred at fixed points, a picture that dates back to the celebrated model of Kronig and Penney \cite{Kronig-Pennig-1931} for a quantum particle in a one-dimensional array of delta potentials.

One can make sense in various conceptually alternative ways of the formal Hamiltonian
\begin{equation}\label{eq:formalHamilt}
^\textrm{``}-\Delta_{x}+\sum_{j=1}^N\,\mu_j\,\delta(x-y_j)^\textrm{''}
\end{equation}
for a quantum particle in $\mathbb{R}^d$ subject to singular interactions centred at the points $y_1,\dots,y_N\in\mathbb{R}^d$ and of magnitude, respectively, $\mu_1,\dots,\mu_N$. One is to realise the Hamiltonian as a self-adjoint extension of the restriction of $-\Delta$ to smooth functions supported away from the $y_j$'s, another is to obtain it as the limit of a Schr\"{o}dinger operator with actual potentials $V^{(j)}_\varepsilon(x-y_j)$ each of which, as $\varepsilon\to 0$, spikes up to a delta-like profile, the support shrinking to the point $\{y_j\}$, and yet another way is to realise \eqref{eq:formalHamilt} as the self-adjoint operator of a closed and semi-bounded energy (quadratic) form that consists of a free (gradient) term plus suitable boundary terms at the centres $y_j$'s.

In the mathematical literature the study of the self-adjoint realisations of \eqref{eq:formalHamilt} has a long history, deeply connected with that of the physical systems for which such a model has provided a realistic description. In this work we focus on the $d=3$ case: we thus consider the collection
\begin{equation}\label{eq:set_Y}
Y\;:=\;\{y_1, \dots, y_N\}
\end{equation}
of $N$ distinct points in $\mathbb{R}^3$ and, correspondingly, the operator
\begin{equation}\label{eq:little_op}
\mathring{H}_Y\;:=\;
-\Delta\upharpoonright C_0^\infty(\mathbb{R}^3 \!\setminus\! Y)
\end{equation}
in the Hilbert space $L^2(\mathbb{R}^3)$. $\mathring{H}_Y$ is densely defined, real symmetric, and non-negative, with deficiency indices 
$(N,N)$, and hence it admits a $N^2$-parameter family of self-adjoint extensions.

The most relevant sub-class of them is the $N$-parameter family
\[
\{H_{\alpha,Y}\,|\,
\alpha\equiv(\alpha_1,\dots,\alpha_N)\in(-\infty,\infty]^N\}
\]
of so-called `local' extensions, namely extensions of $\mathring{H}_Y$ whose domain of self-adjointness is only qualified by certain local boundary conditions at the singularity centres. More precisely, as we shall recall in detail in Section \ref{sec:preliminaries}, the domain $\mathcal{D}(H_{\alpha,Y})$ of $H_{\alpha,Y}$ consists of functions $u$ satisfying the asymptotics
\begin{equation} \label{bdry} 
\lim_{r_j\downarrow 0} 
\Big(\frac{\partial(r_j u)}{\partial r_j}- 4\pi \alpha_j r_j u\Big) =0\,,\quad r_j:=|x-y_j|\,, \quad j\in\{1, \dots, N\}\,. 
\end{equation} 
In fact, the condition
\begin{equation}\label{eq:BPcontact}
u(x)\;\underset{x\to y_j}{\sim}\;\frac{1}{\,|x-y_j|\,}-\frac{1}{\,a_j}\,,\qquad a_j:=-(4\pi\alpha_j)^{-1}
\end{equation}
implied by \eqref{bdry} is typical for the low-energy behaviour of an eigenstate of the Schr\"{o}dinger equation for a quantum particle subject to a potential of extremely short, virtually zero, range centred at the point $y_j$ and with $s$-wave scattering length $a_j$, a fact that was noted first by Bethe and Peierls \cite{Bethe_Peierls-1935,Bethe_Peierls-1935-np} (whence the name of \emph{Bethe-Peierls contact condition} for the asymptotics \eqref{eq:BPcontact}).

If for some $j\in\{1,\dots,N\}$ one has $\alpha_j=\infty$, then no actual interaction is present at the point $y_j$ (no boundary condition as $x\to y_j$) and in practice things are as if one discards the point $y_j$. In particular, the extension $H_{\alpha,Y}$ corresponding to $\alpha=\infty$ is the Friedrichs extensions of $\mathring{H}_Y$, namely the self-adjoint negative Laplacian on $L^2(\mathbb{R}^3)$. We shall also denote it by $H_0$, the free Hamiltonian. It is precisely the extension with no interactions at all.

The operator $H_{\alpha,Y}$ was rigorously studied for the first 
time by Albeverio, Fenstad, and H\o{}egh-Krohn 
\cite{Albeverio-Fenstad-HoeghKrohn-1979_singPert_NonstAnal} 
and subsequently characterised by Zorbas \cite{Zorbas-1980}, 
Grossmaann, H\o{}egh-Krohn, and Mebkhout 
\cite{Grossmann-HK-Mebkhout-1980,Grossmann-HK-Mebkhout-1980_CMPperiodic}, and D{\setbox0=\hbox{a}{\ooalign{\hidewidth\lower1.5ex\hbox{`}\hidewidth\crcr\unhbox0}}}browski and Grosse \cite{Dabrowski-Grosse-1985}. A thorough discussion of its features can be found in \cite[Section II.1.1]{albeverio-solvable}. 
We shall recall the main properties of $H_{\alpha,Y}$ in 
Section \ref{sec:preliminaries}.

Fundamental information about the dynamics generated by 
$H_{\alpha,Y}$ through the Schr\"{o}dinger equation 
$\ii\partial_t u=H_{\alpha,Y}u$ is encoded in the wave 
operators for the pair $(H_{\alpha,Y}, H_0)$, which are defined 
by the strong limits
\begin{equation}
W^{\pm}_{\alpha, Y}\;=\;\lim_{t\to\pm \infty} 
e^{\ii tH_{\alpha,Y}}e^{-\ii tH_0}\,.
\end{equation}

Since the resolvent difference 
$(H_{\alpha,Y} -z^2{\mathbbm{1}})^{-1} -(H_0-z^2{\mathbbm{1}})^{-1} $ is 
of finite rank, as we shall recall in Theorem \ref{thm:general_properties} (see \eqref{eq:resolvent_identity} below), 
standard arguments from scattering theory \cite{rs3} 
guarantee that the wave operators $W^{+}_{\alpha, Y}$ and 
$W^{-}_{\alpha, Y}$ exist in $L^2(\mathbb{R}^3)$ and are 
complete, meaning that 
\begin{equation}
 \mathrm{ran}\,W^{\pm}_{\alpha, Y}\;
=\;L^2_{\mathrm{ac}}(H_{\alpha,Y})\;
=\;P_{\mathrm{ac}}(H_{\alpha,Y})L^2(\mathbb{R}^3)\,,
\end{equation}
where $L^2_{\mathrm{ac}}(H_{\alpha,Y})$ denotes the  
absolutely continuous spectral subspace of $\mathcal{H}$ for 
$H_{\alpha,Y}$, and $P_{\mathrm{ac}}(H_{\alpha,Y})$ denotes the orthogonal 
projection onto $L^2_{\mathrm{ac}}(H_{\alpha,Y})$. In particular, 
the absolutely continuous part of $H_{\alpha,Y}$, namely the operator
$H_{\alpha,Y}P_{\mathrm{ac}}(H_{\alpha,Y})$,
is unitarily equivalent to $H_0$. Moreover, the singular continuous 
spectrum is absent from $H_{\alpha,Y}$ and the point spectrum 
consists of at most $N$ negative 
eigenvalues, whereas non-negative eigenvalues are absent
(see Theorem \ref{thm:general_properties} below).

Wave operators are of paramount importance for the study of the scattering governed by 
an interaction Hamiltonian in comparison with a free (reference) Hamiltonian 
\cite{Kuroda1978_intro_scatt_theory,rs3}. Owing to their completeness, 
$W^{+}_{\alpha, Y}$ and $W^{-}_{\alpha, Y}$ are unitary from $L^2(\mathbb{R}^3)$ onto 
$L^2_{\mathrm{ac}}(H_{\alpha,Y})$; moreover, they intertwine 
$H_{\alpha,Y}P_{\mathrm{ac}}(H_{\alpha,Y})$ and $H_0$, viz., for any Borel function $f$ on $\mathbb{R}$ one has the identity
\begin{equation}\label{eq:intertwining}
f(H_{\alpha,Y})P_{\mathrm{ac}}(H_{\alpha,Y})\;=\; 
W^{\pm}_{\alpha, Y} \,f(H_0) \,(W^{\pm}_{\alpha, Y})^*\,.
\end{equation} 
Through this intertwining, mapping properties of $f(H_{\alpha,Y})P_{\mathrm{ac}}(H_{\alpha,Y})$ can be deduced from those of $f(H_0)$ (which, upon Fourier transform, is the multiplication by $f(\xi^2)$), provided that the corresponding ones of $W^{\pm}_{\alpha, Y}$ are known. Thus, the $L^{p'}\to L^p$ boundedness of $f(H_{\alpha,Y})P_{\mathrm{ac}}(H_{\alpha,Y})$ follows from the $L^p\to L^p$ boundedness of $W^{\pm}_{\alpha, Y}$: more precisely, if $W^{\pm}_{\alpha, Y}\in\mathcal{B}(L^p(\mathbb{R}^d))$ for some $p\in[1,\infty]$, then $(W^{\pm}_{\alpha, Y})^*\in\mathcal{B}(L^{p'}(\mathbb{R}^d))$ and hence
\begin{equation}\label{eq:p-boundedness_lifted}
\|f(H_{\alpha,Y})P_{\mathrm{ac}}(H_{\alpha,Y})\|_{\mathcal{B}(L^{p'},L^p)}\;\leqslant\;C_{p}\,\|f(H_0)\|_{\mathcal{B}(L^{p'},L^p)}\,,
\end{equation}
the constant $C_{p}$ being \emph{independent} of $f$.
(Here and henceforth, $p'$ will denote the conjugate of $p$ via $p^{-1}+{p'}^{-1}=1$.)

The literature on the $L^p$-boundedness of wave operators relative 
to actual Schr\"{o}d\-inger operators of the form $-\Delta+V$, 
for sufficiently regular $V:\mathbb{R}^d\to\mathbb{R}$ vanishing at 
spatial infinity, is vast \cite{Yajima-JMSJ-1995-1,Yajima-CMP-1999-2D,Artbazar-Yajima_JMSUT-2000,Weder-JFA-2000,Jensen-Yajima-CMP-2002,DAncona-Fanelli-CMP-2006,Finco-Yajima-JMSUT-2006,Jensen-Yajima-ProcLondMathSoc-2008,Beceanu-AJM-2014,Yajima-DocMath2016,Yajma2016_3D_WaveOps_ThreshSing,Beceanu-Schlag-2016,Beceanu-Schlag-2017} and the problem is well known to depend crucially on the spectral properties of $-\Delta+V$ at the bottom of the absolutely continuous spectrum, that is, at energy zero.

For \emph{singular perturbations} of the Schr\"{o}dinger operators, the picture is much less developed and is essentially limited to the \emph{one-dimensional} case. Analogously to \eqref{eq:little_op}, the restriction $\mathring{H}_Y:=-\Delta\upharpoonright C_0^\infty(\mathbb{R} \!\setminus\! Y)$ admits a $N^2$-parameter family of self-adjoint extensions in $L^2(\mathbb{R})$ \cite[Section II.2.1]{albeverio-solvable}, among which those extensions with so-called separated boundary condition of $\delta$-type -- the analogue of $H_{\alpha,Y}$ considered above. For the latter sub-family of Hamiltonians Duch\^{e}ne, Marzuola, and Weinstein \cite{Duchene-Marzuola-Weinstein-2010} constructed the corresponding wave operators $W^{\pm}_{\alpha, Y}$ relative to the couple $(H_{\alpha,Y},H_0)$ and proved that $W^{\pm}_{\alpha, Y}\in\mathcal{B}(W^{1,p}(\mathbb{R}))$ for $1<p<\infty$. Their proof is built on a detailed  decomposition of $W^{\pm}_{\alpha, Y}$ essentially based upon the high frequency vs low frequency behaviour of the Jost solutions, an eminently one-dimensional treatment that is hard to export to higher dimensions.

In this work we study $L^p$-bounds for the wave operators 
$W^{\pm}_{\alpha, Y}$ of the \emph{three-dimensional} multi-centre point interaction Hamiltonian. We provide a manageable formula for (the integral kernel of) $W^{\pm}_{\alpha, Y}$, which we obtain by manipulating the resolvent difference $(H_{\alpha,Y}-z^2{{\mathbbm{1}}})^{-1}-(H_0-z^2{{\mathbbm{1}}})^{-1}$: since this difference is an explicitly known finite rank operator for any dimensions $d=1,2,3$, our derivation can be naturally exported also to lower dimensions.

Based on our representation of $W^{\pm}_{\alpha, Y}$, we then establish our main result:

\begin{thm}\label{ref:main_thm} 
For any $y_1, \dots, y_N \in \mathbb{R}^3$ 
and $\alpha_1, \dots, \alpha_N\in \mathbb{R}$, the wave operators
\begin{equation}\label{eq:main_thm_W_op}
W^{\pm}_{\alpha, Y}\;=\;s\textrm{-}\!\!\!\!\lim_{t\to\pm \infty} e^{\ii tH_{\alpha,Y}}e^{-\ii tH_0}
\end{equation}
for the pair $(H_{\alpha,Y}, H_0)$ exist and are complete in $L^2(\mathbb{R}^3)$, they are \emph{bounded} 
in $L^p(\mathbb{R}^3)$ for $1<p<3$, and \emph{unbounded} for $p=1$ and for 
$p\geqslant 3$.
\end{thm}

\begin{rem}\label{re:consistent}
The fact that $L^p$-boundedness holds only for $p\in(1,3)$ is consistent with the analogous result for actual Schr\"{o}dinger operators. Indeed, it is well known \cite{Yajima-DocMath2016,Yajma2016_3D_WaveOps_ThreshSing} that 
the wave operators for three-dimensional Schr\"{o}dinger 
operators $-\Delta+V$ admitting a zero-energy resonance 
are $L^p$-bounded  if and only if
$p\in(1,3)$, and moreover it is also well known \cite[Theorem II.1.2.1]{albeverio-solvable} that $H_{\alpha,Y}$ is actually the strong resolvent 
limit in $L^2(\mathbb{R}^3)$, as $\varepsilon\downarrow 0$, 
of Schr\"odinger operators of the form 
\begin{equation}\label{eq:scaled_approx}
H^{(\varepsilon)}\;=\; -\Delta + 
\varepsilon^{-2} \sum_{j=1}^N \lambda_j(\varepsilon)V_j\Big(\frac{x-y_j}{\varepsilon}\Big)
\end{equation}
for suitable real-analytic $\lambda_j(\varepsilon)$'s with 
$\lambda(0)=1$ and real potentials $V_j$ of finite Rollnik norm such that
$-\Delta +V_j$ has a zero-energy resonance for each 
$j\in\{1,\dots,N\}$.
To fully substantiate such a parallelism between singular and regular Schr\"{o}dinger operators, it would be of great interest to monitor the convergence, as bounded operators in $L^p(\mathbb{R}^3)$ for $p\in(1,3)$, of the wave operators for the pair $(H_{\varepsilon},H_{0})$ to the wave operator $W_{\alpha,Y}^{\pm}$. Along this line, in Section \ref{sec:convergence_of_wave} we present the proof of this result in the special case $N=1$, $\alpha=0$.
\end{rem}

As a direct consequence of Theorem \ref{ref:main_thm} and of the bound \eqref{eq:p-boundedness_lifted}, the dispersive properties for the free propagator $e^{-\ii tH_0}$, encoded in the estimates
\begin{equation}\label{eq:p-q-f}
\|e^{-\ii tH_0}u\|_p\;\leqslant\;
(4{\pi}|t|)^{-3(\frac12-\frac1{p})}\|u\|_{p'}, \qquad p\in[2,+\infty]\,,\quad t\neq 0\,, 
\end{equation} 
lift to analogous estimates for the Schr\"{o}dinger dynamics generated by $H_{\alpha,Y}$, albeit for an unavoidably smaller range of $p$'s than in \eqref{eq:p-q-f}. Thus, we find:

\begin{cor}\label{cor:cor-1}
There is a constant $C>0$ such that, for each $p\in[2,3)$,
\begin{equation}\label{eq:p-q}
\|e^{-\ii tH_{\alpha,Y}}P_{\rm ac}(H_{\alpha,Y})u\|_p\;\leqslant\;C\,
|t|^{-3(\frac12-\frac1{p})}\|u\|_{p'}, \qquad  t\neq 0\,.
\end{equation}
\end{cor}

In turn, by means of a well-known argument 
\cite{GinibreVelo_GWP_NLS_revisited_1985,
Yajima1987_existence_soll_SE}, the dispersive estimates 
\eqref{eq:p-q} imply Strichartz estimates for $H_{\alpha,Y}$ 
for the same range of $p$.
We shall call a pair of exponents 
$(p,q)$ \emph{admissible for} $H_{\alpha,Y}$ if 
\begin{equation}\label{eq:admie}
p\in[2,3)\qquad\textrm{and}\qquad 0\;\leqslant\;\frac{2}{q}\;=\;3\,\Big(\frac{1}{2}-\frac{1}{p}\Big)\;<\;\frac{1}{2}\,,
\end{equation}
that is, $q=\frac{4p}{3(p-2)}\in(4,+\infty]$. 

\begin{cor} Let  $(p,q)$ and $(r,s)$ be two admissible pairs 
for $H_{\alpha,Y}$.
Then, for a 
constant $C>0$, 
\begin{equation} \label{stri-1}
\|e^{-\ii tH_{\alpha,Y}}P_{\mathrm{ac}}(H_{\alpha,Y})u\|_{L^q(\mathbb{R}_t, L^p(\mathbb{R}_x^3))}\;\leqslant\;
C\|u\|_{L^2(\mathbb{R}^3)}
\end{equation} 
and  
\begin{equation} \label{stri-2}
\left\| 
\int_0^t e^{-\ii(t-s)H_{\alpha,Y}}P_{\mathrm{ac}}(H_{\alpha,Y}) \,u(s)\,\ud s
\right\|_{L^q(\mathbb{R}_t, L^p(\mathbb{R}^3_x))}
\;\leqslant\;
C\|u\|_{L^{s'}(\mathbb{R}_t, L^{r'}(\mathbb{R}^3_x))}\,.
\end{equation} 
\end{cor}

Under the \emph{additional} assumption that the matrix $\Gamma_{\alpha,Y}(\lambda)$ that we define in equation \eqref{ga-def} in Section \ref{sec:preliminaries} be invertible for all $\lambda\in[0,+\infty)$, with locally bounded inverse, suitably weighted dispersive estimates for the propagator $e^{-\ii t H_{\alpha,Y}}$ were obtained by D'Ancona, Pierfelice, and Teta \cite{Dancona-Pierfelice-Teta-2006} in the form
\begin{equation}\label{eq:DPT-dispersive}
\|w^{-1} 
e^{-\ii t H_{\alpha,Y}}P_{\mathrm{ac}}(H_{\alpha,Y})u\|_\infty\;\leqslant\; 
C\,|t|^{-3/2}
\|w\,u\|_{1}\,,\qquad t\neq 0 \,,
\end{equation} 
for the weight function
\begin{equation}
w(x)\;:=\;\sum_{j=1}^N \left(1+ \frac1{|x-y_j|}\right)\,. 
\end{equation}
The restriction on $\Gamma_{\alpha,Y}(\lambda)$ is in practice the requirement that zero is not a resonance for $H_{\alpha,Y}$; thus, for $N=1$, \eqref{eq:DPT-dispersive} was proved for $\alpha\neq 0$ and it was replaced by a slower dispersion rate $|t|^{-1/2}$ in the resonant case $N=1$, $\alpha=0$. We also observe that by interpolation \eqref{eq:DPT-dispersive} can be turned into the weighted dispersive estimate
\begin{equation}\label{eq:p-q-d}
\|w^{-\left(1-\frac2{p}\right)} 
e^{-\ii t H_{\alpha,Y}}P_{\mathrm{ac}}(H_{\alpha,Y})u\|_p\;\leqslant\; 
C_p\,|t|^{-3(\frac12-\frac1{p})}
\|w^{\frac{2}{p'}-1}u\|_{p'}
\end{equation} 
for the \emph{whole} range $p\in[2,+\infty]$.

As opposite to \eqref{eq:p-q-d}, our Corollary \ref{cor:cor-1} removes both the weight and the assumption on $\Gamma_{\alpha,Y}(\lambda)$ in the regime $p\in[2,3)$. In fact, we can also improve the weight in \eqref{eq:p-q-d} for $p\in[3,+\infty]$ by interpolating between \eqref{eq:p-q} of our Corollary \ref{cor:cor-1} and \eqref{eq:DPT-dispersive} given by \cite{Dancona-Pierfelice-Teta-2006}.

We also highlight that in the parallel work \cite{Iandoli-Scandone-2017} by one of us in collaboration with Iandoli, the non-weighted dispersive estimate \eqref{eq:p-q} is recovered by simpler and more direct arguments (i.e., without using any result from the scattering theory for $H_{\alpha,Y}$) in the special case $N=1$.

The first key ingredient of our analysis is a fairly explicit resolvent formula for $H_{\alpha,Y}$, which is well known to be a rank-$N$ perturbation (in the resolvent sense) of the free Hamiltonian. This is in a way the same spirit as in the above-mentioned work \cite{Dancona-Pierfelice-Teta-2006} for generic $N$, except that the main difficulty therein was to produce reliable estimates on the propagator $e^{-\ii t H_{\alpha,Y}}$ in the lack of an explicit representation of its kernel (instead, when $N=1$ the dispersive estimate \eqref{eq:DPT-dispersive} was obtained in \cite{Dancona-Pierfelice-Teta-2006} directly from the explicit kernel of the propagator $e^{-\ii t H_{\alpha,Y}}$, a kernel found by Scarlatti and Teta \cite{Scarlatti-Teta-1990} and by Albeverio, Brze\'zniak, and D{\setbox0=\hbox{a}{\ooalign{\hidewidth\lower1.5ex\hbox{`}\hidewidth\crcr\unhbox0}}}browski \cite{Albeverio_Brzesniak-Dabrowski-1995}).

In our case we aim at representing the (kernel of the) wave operators $W^{\pm}_{\alpha, Y}$ in the first place, based on the explicit resolvent difference $(H_{\alpha,Y}-z^2{\mathbbm{1}})^{-1}-(H_{0}-z^2{\mathbbm{1}})^{-1}$. Then, as a second key ingredient, for the $L^p\to L^p$ estimate of $W^{\pm}_{\alpha, Y}$ we appeal to a large extent to some tool from harmonic analysis, the Calder\'on-Zygmund operators and the Muckenhaupt weighted inequalities.

We organised the material as follows. 
In Section \ref{sec:preliminaries} we recall the precise definition of 
$H_{\alpha, Y}$ and we collect several technical results needed in the proof of Theorem \ref{ref:main_thm}, including in particular properties of Calder\'on-Zygmund operators and the Muckenhaupt weighted inequalities. In Section \ref{sec:Stationary_Repr_W_ops} we produce the explicit stationary representation of the wave operators $W^{\pm}_{\alpha, Y}$ which the proof of Theorem \ref{ref:main_thm} is based on. The 
$L^p$-boundedness part of Theorem \ref{ref:main_thm} is proved in 
Section \ref{sec:Lp-bounds} for the single centre case, and in 
Section \ref{sec:Lp-bounds_N} for the multi-centre case. The 
$L^p$-unboundedness part is proved in Section \ref{sec:unbound}. 
Last, in Section \ref{sec:convergence_of_wave} we discuss the 
convergence of the wave operators relative to the family of 
Hamiltonians \eqref{eq:scaled_approx} to the wave operators 
$W^{\pm}_{\alpha, Y}$ (limit of shrinking potentials).

\section{Preliminaries and 
notation}\label{sec:preliminaries}  

In this Section we recall the precise definition of 
$H_{\alpha,Y}$ and its basic properties from  \cite[Section II.1.1]{albeverio-solvable} and \cite{Posilicano2000_Krein-like_formula} (see also \cite{DFT-brief_review_2008, Dancona-Pierfelice-Teta-2006}). Here and henceforth the number $N\in\mathbb{N}$ and the $N$-point set $Y=\{y_1,\dots,y_N\}$ introduced in \eqref{eq:set_Y} are fixed, and the multi-dimensional parameter $\alpha\equiv(\alpha_1,\dots,\alpha_N)$ is assumed to run over $(-\infty,+\infty)^N$.

We begin with a few remarks on our notation.
We write $\mathbb{C}$ for the complex plane and $\mathbb{C}^{+}$ for the open 
upper half plane. By $\delta_{j,\ell}$ we denote the Kronecker delta, namely the quantity 1 for $j=\ell$ and 0 otherwise. As customary, $\langle\lambda\rangle\equiv(1+\lambda^2)^{\frac{1}{2}}$ for $\lambda\in\mathbb{R}$. The representation of any point $x\in\mathbb{R}^3$ in polar coordinates will be $x=r\omega$, where $r\equiv|x|\geqslant 0$ and $\omega\in\mathbb{S}^2$.
For $u, v\in L^2(\mathbb{R}^3)$, 
$u\otimes v$ denotes the rank-1 operator 
$f \mapsto u\langle v, f\rangle$, where $\langle\cdot,\cdot\rangle$ is the usual scalar product in $L^2(\mathbb{R}^3)$, anti-linear in the first entry and linear in the second. 
For the Fourier transform in $\mathbb{R}^d$ we use the convention
\[
(\mathcal{F}f)(\xi)\;\equiv\;\widehat{f}(\xi)\;=\;\frac{1}{\;(2\pi)^{d/2}}\int_{\mathbb{R}^d}e^{-\ii x\xi}f(x)\,\ud x\,.
\]
We often write $f\leqslant_{|\,\cdot\,|} g$ when $|f|\leqslant |g|$. 
$E^{(T)}(\ud \lambda)$ denotes the spectral measure of the self-adjoint operator $T$.
When not specified otherwise, $C$ denotes a universal positive constant and ${\mathbbm{1}}$ is the identity operator on the space that is clear from the context (this includes also the case of the $N\times N$ identity matrix).

For $z\in \mathbb{C}$ and $x,y,y'\in \mathbb{R}^3$, we set 
\begin{equation}\label{eq:def_of_the_Gs}
\begin{split}
\mathcal{G}_z(x)\;&:=\;\frac{e^{\ii z|x|}}{\,4\pi |x|\,}\,,\qquad \mathcal{G}_z^y(x)\;:=\;\frac{e^{\ii z|x-y|}}{\,4\pi |x-y|\,}\;=\;\mathcal{G}_z(x-y)\,, \\
\mathcal{G}_z^{yy'}\;&:=\;\begin{cases}
\displaystyle\frac{e^{\ii z|y-y'|}}{\,4\pi |y-y'|\,} & \textrm{if }\;y'\neq y \\ 
\qquad 0 & \textrm{if }\;y'= y\,, \\
\end{cases}
\end{split}
\end{equation}
and 
\begin{equation}\label{ga-def}
\Gamma_{\alpha,Y}(z)\;:=\;\Big(\Big(\alpha_j-\frac{\ii z}{\,4\pi\,}\Big)\delta_{j,\ell}-\mathcal{G}_z^{y_jy_\ell}\Big)_{\!j,\ell=1,\dots,N}\,.
\end{equation}
Thus, the function $z\mapsto \Gamma_{\alpha,Y}(z)$ has values in the $N\times N$ symmetric matrices and is clearly entire, and $z\mapsto \Gamma_{\alpha,Y}(z)^{-1}$ is meromorphic in $z\in\mathbb{C}$.

It is known that $\Gamma_{\alpha,Y}(z)^{-1}$ has at most $N$ poles in the closed upper half plane $\mathbb{C}^{+}\cup \mathbb{R}$, which are all located along the positive imaginary semi-axis
\cite[Theorem II.1.1.4]{albeverio-solvable}. We denote by $\mathcal{E}$ the set of such poles.

The following facts are known.

\begin{thm}\label{thm:general_properties}~
\begin{itemize}
 \item[(i)] For $z\in\mathbb{C}^+\!\setminus\!\mathcal{E}$ 
the identity
\begin{equation}\label{eq:resolvent_identity}
(H_{\alpha,Y} -z^2{\mathbbm{1}})^{-1} -(H_0-z^2{\mathbbm{1}})^{-1} \;=\; 
\sum_{j,k=1}^N  (\Gamma_{\alpha,Y}(z)^{-1})_{jk} \,\mathcal{G}_{z}^{y_j}\otimes   
\overline{\mathcal{G}_{z}^{y_k}}
\end{equation} 
defines the resolvent of a self-adjoint operator $H_{\alpha,Y}$ in $L^2(\mathbb{R}^3)$. $H_{\alpha,Y}$ is an extension of the operator $\mathring{H}_Y=-\Delta\upharpoonright C_0^\infty(\mathbb{R}^3 \!\setminus\! Y)$ defined in \eqref{eq:little_op}.
\item[(ii)] The domain of $H_{\alpha,Y}$ has the following representation, for any $z\in\mathbb{C}^+\!\setminus\!\mathcal{E}$:
\begin{equation}\label{eq:domain_of_HaY}
\mathcal{D}(H_{\alpha,Y})\;=\;\Big\{\psi= \phi_z + \sum_{j,k=1}^N (\Gamma_{\alpha,Y}(z)^{-1})_{jk} \, \phi_z(y_k)
{\mathcal{G}}_{z}^{y_j}  
\,\Big|\,\phi_z \in H^2(\mathbb{R}^3) \Big\}\,.
\end{equation} 
The summands in the decomposition of each $\psi\in\mathcal{D}(H_{\alpha,Y})$ depend on the chosen $z$, however, $\mathcal{D}(H_{\alpha,Y})$ does not. Equivalently, for any $z\in\mathbb{C}^+\!\setminus\!\mathcal{E}$,
\begin{equation}\label{eq:domain_of_HaY_bc}
\mathcal{D}(H_{\alpha,Y})\;=\;\left\{\psi= \phi_z +\sum_{j=1}^N q_j\,{\mathcal{G}}_{z}^{y_j}\left|\!\!
\begin{array}{c}
\phi_z \in H^2(\mathbb{R}^3) \\
(q_1,\dots,q_N)\in\mathbb{C}^N \\
\begin{pmatrix} \phi_z(y_1) \\ \vdots \\ \phi_z(y_N)\end{pmatrix}=\Gamma_{\alpha,Y}(z)\begin{pmatrix} q_1 \\ \vdots \\ q_N\end{pmatrix}
\end{array}\!\!\!\right.\right\}\,.
\end{equation}
At fixed $z$, the decompositions above are unique.
\item[(iii)] With respect to the decompositions \eqref{eq:domain_of_HaY}-\eqref{eq:domain_of_HaY_bc}, one has
\begin{equation}\label{eq:action_of_H}
 (H_{\alpha,Y} -z^2{\mathbbm{1}})\,\psi\;=\;(H_0 -z^2{\mathbbm{1}})\,\phi_z\,.
\end{equation}
Moreover, $H_{\alpha,Y}$ has the following locality property: if 
$\psi\in\mathcal{D}(H_{\alpha,Y})$ is such that 
$\psi|_\mathcal{U}\equiv 0$ for some open 
$\mathcal{U}\subset\mathbb{R}^3$, then 
$(H_{\alpha,Y}\psi)|_\mathcal{U}\equiv 0$.
\item[(iv)] The spectrum $\sigma(H_{\alpha,Y})$ of 
$H_{\alpha,Y}$ consists of at most 
$N$ strictly negative eigenvalues and the absolutely continuous 
part $\sigma_{\mathrm{ac}}(H_{\alpha,Y})=[0,\infty)$. 
Non-negative eigenvalues and the singular continuous 
spectrum are absent. 
There is a one to one correspondence between 
the poles $\ii\lambda$, $\lambda>0$ of 
$\Gamma_{\alpha,Y}(z)^{-1}$ in $\mathbb{C}^{+}$ and the 
negative eigenvalues $-\lambda^2$ of $H_{\alpha,Y}$, counting the 
multiplicity. The eigenfunctions that belong to 
the eigenvalue $E_0=-\lambda_0^2<0$ of $H_{\alpha,Y}$ have the form
\[
 \psi_0\;=\;\sum_{j=1}^N c_j\,\mathcal{G}_{\ii\lambda_0}^{y_j}
\]
where $(c_0,\dots,c_N)$ are eigenvectors with eigenvalue zero of $\Gamma_{\alpha,Y}(\ii\lambda_0)$. The ground state, if it exists, is non-degenerate.
\end{itemize}
\end{thm}

Part (i) of Theorem \ref{thm:general_properties} above was first proved in  \cite{Grossmann-HK-Mebkhout-1980,Grossmann-HK-Mebkhout-1980_CMPperiodic} -- see also the discussion in \cite[equation (II.1.1.33)]{albeverio-solvable}.
Parts (ii) and (iii) originate from \cite{Grossmann-HK-Mebkhout-1980_CMPperiodic} and are discussed in \cite[Theorem II.1.1.3]{albeverio-solvable}, in particular \eqref{eq:domain_of_HaY_bc} is highlighted in
\cite{DFT-brief_review_2008}. Part (iv) is an extension, proved in \cite[Theorem II.1.1.4]{albeverio-solvable}, of some of the corresponding results established in \cite{Grossmann-HK-Mebkhout-1980_CMPperiodic}.

By exploiting the boundary condition \eqref{eq:domain_of_HaY_bc} between the regular and the singular part of a generic $\psi\in\mathcal{D}(H_{\alpha,Y})$, it is straightforward to see that
\begin{equation}\label{eq:bc_condition_domain}
\lim_{r_j\downarrow 0} 
\Big(\frac{\partial(r_j \psi)}{\partial r_j}- 4\pi \alpha_j r_j \psi\Big) =\,0\,,\quad 
r_j:=|x-y_j|\,, \quad j\in\{1, \dots, N\}\,,
\end{equation} 
whence also
\begin{equation}\label{eq:asymptotics_of_psi_in_domain}
\lim_{x\to y_j} \Big( 
\psi(x) - \frac{q_j}{4\pi|x-y_j|}  - \alpha_j q_j \Big)\;=\;0 , \qquad 
j\in\{1, \dots, N\}\,.
\end{equation}
Thus, the elements of $\mathcal{D}(H_{\alpha,Y})$ satisfy the `physical' (Bethe-Peierls) boundary condition
\begin{equation}
\psi(x)\;\underset{x\to y_j}{\sim}\;\frac{q_j}{\,4\pi\,}\Big(\frac{1}{\,|x-y_j|\,}-\frac{1}{\,a_j}\Big)\,,\qquad a_j:=-(4\pi\alpha_j)^{-1}
\end{equation}
at each centre of the point interaction (see \eqref{eq:BPcontact} in the Introduction). In fact, $\mathcal{D}(H_{\alpha,Y})$ is nothing but the space of those $L^2$-functions $\psi$ such that the distribution $\Delta\psi$ belongs to $L^2(\mathbb{R}^3\!\setminus\!Y)$ and the boundary condition \eqref{eq:bc_condition_domain} is satisfied.

We also record two simple consequences of Theorem \ref{thm:general_properties} which will turn out to be useful in our discussion.

\begin{lem} \label{reality} 
The operator $H_{\alpha,Y}$ is a \emph{real} self-adjoint 
operator, that is, for a real-valued function $\psi\in\mathcal{D}(H_{\alpha,Y})$, 
$H_{\alpha,Y}\psi$ is also real-valued.  
\end{lem}

\begin{proof} Let $z=\ii\lambda$, $\lambda>0$, be such that 
$\ii\lambda \not\in \mathcal{E}$ and let $\psi$ be a real-valued function in $\mathcal{D}(H_{\alpha,Y})$. 
Then, with the notation of the decomposition \eqref{eq:domain_of_HaY_bc} of $\psi$, the asymptotics \eqref{eq:asymptotics_of_psi_in_domain} show that the coefficients $q_1,\dots,q_N$ are all real.
The entries of $\Gamma_{\alpha,Y}(\ii\lambda)$ are real too, because $\mathfrak{Re}\,z>0$. Then \eqref{eq:domain_of_HaY_bc} implies that $\phi_z$ is real-valued and so must be $H_{\alpha,Y}\psi+\lambda^2\psi$, owing to \eqref{eq:action_of_H}.
\end{proof}

\begin{lem}\label{lem:pole_of_order_1} If $z=0$ is a pole 
of $\Gamma_{\alpha,Y}(z)^{-1}$, then it is a pole of  
first order and in a neighbourhood of $z=0$ one has 
\[
\Gamma_{\alpha,Y}(z)^{-1}\;=\;\frac{\Theta}{z} + \Gamma_{\alpha,Y}^{\mathrm{(reg)}}(z)
\]
for some constant matrix $\Theta$ and some analytic matrix-valued function $\Gamma_{\alpha,Y}^{\mathrm{(reg)}}(z)$. 
\end{lem}

\begin{proof} We recall first that for a generic self-adjoint operator $T$ in a Hilbert space
$\mathcal{H}$ for which zero is not an eigenvalue, one has 
\[
\begin{split}
\lim_{\k\downarrow 0 }\|\:\ii \k^2(T+\ii\k^2)^{-1}u\|^2\;&=\;\lim_{\k\downarrow 0 } \int_{\mathbb{R}}
\left|\frac{\ii\k^2}{\lambda+\ii\k^2}\right|^2\langle u,E^{(T)}(\ud \lambda)u\rangle \\
&=\;\|E^{(T)}(\{0\})u\|^2 \;=\; 0\qquad\forall u\in\cH\,.
\end{split}
\]
In our case, neither $H_{\a,Y}$ nor $H_0$ have zero eigenvalue: therefore, applying the above fact to the resolvent identity \eqref{eq:resolvent_identity}, one finds
\begin{equation*} 
\lim_{z\to 0} \,\sum_{j,k=1}^N z^2 (\Gamma_{\alpha,Y}(z)^{-1})_{jk} 
\langle u, \mathcal{G}_{z}^{y_j}\rangle\,\langle \mathcal{G}_{-\ii\overline{z}}^{y_k},v\rangle \;=\;0
\end{equation*} 
for any $u, v \in C_0^\infty(\mathbb{R}^3)$. Suppose 
that $\Gamma_{\alpha,Y}(z)^{-1}$ has a pole of order $\geqslant 2$ at $z=0$ with matrix residue $\widetilde{\Theta}$: then the identity above implies
\[
\sum_{j,k=1}^N \widetilde{\Theta}_{jk}\langle u, \mathcal{G}_{0}^{y_j}\rangle\,\langle \mathcal{G}_{0}^{y_k},v\rangle \;=\;0\,\qquad\forall u, v \in C_0^\infty(\mathbb{R}^3)\,.
\] 
It follows that 
\[
\sum_{j,k=1}^N\, \frac{\widetilde{\Theta}_{jk}}{\,|x-y_j|\,|y-y_k|\,} \;=\;0\,, \qquad x,y \in \mathbb{R}^3.
\]
Since $\widetilde{\Theta}$ is a symmetric matrix, this implies 
$\widetilde{\Theta}=0$ and the pole must be of first order. 
\end{proof}

Last, we collect in the remaining part of this Section some results from one-dimensional harmonic analysis which we shall make crucial use of in the course of our discussion. For the definition of 
Calder\'on-Zygmund operators 
we refer to 
\cite[Definitions 7.4.1, 7.4.2]{Grafakos_ClassialFourier} 
and to 
\cite[Definitions 4.1.2 and 4.1.8]{Grafakos_ModernFourier}, 
whereas for the definition of $A_p$ Muckenhaupt weights we refer to 
\cite[Definitions 7.1.3]{Grafakos_ClassialFourier}. 
We shall use interchangeably the same symbol for a Calder\'on-Zygmund operator and for its integral kernel.

The following properties are known.

\begin{thm}\label{thm:harmonic_analysis_results}~
\begin{itemize}
 \item[(i)] The convolution operator on $\mathbb{R}$ with 
a function $L(x)$ is a Calder\'on-Zygmund operator if 
$\widehat{L}(\xi)$ is bounded and, for a constant $C>0$, one has  
\begin{equation*} 
|L(x)|\;\leqslant\; C\,|x|^{-1} \qquad\textrm{and}\qquad
\Big|\frac{\ud L}{\ud x}(x)\Big|\;\leqslant\; C \,|x|^{-2} \qquad\textrm{for}\;
x\not=0\, .
\end{equation*} 
 \item[(ii)] If $L$ is a Calder\'on-Zygmund operator 
and $w$ is an $A_p$-weight for some $p\in(1,\infty)$, then 
$L$ is bounded in $L^p(\mathbb{R}, w(x)dx)$ in the sense that
\begin{equation}\label{eq:boundedness_of_CZ_ops_with_Apweghts}
\int_{\mathbb{R}} |(Lu)(x)|^p\, w(x)\,\ud x \;\leqslant\; \int_{\mathbb{R}} |u(x)|^p \,w(x) \,\ud x \qquad\forall u \in C_0^\infty(\mathbb{R})\,.
\end{equation}
\item[(iii)] If $w$ is an $A_p$-weight for some $p\in(1,\infty)$ and 
\begin{equation}
(\mathcal{M}(u))(x) \;:=\; \sup_{r>0}\,\frac{1}{2r}\int_{|x-y|<r} 
|u(y)|\,\ud y
\end{equation}
is the Hardy-Littlewood maximal function of some $u \in C_0^\infty(\mathbb{R})$, then 
\begin{equation}\label{eq:boundedness_of_HL_with_Apweghts}
\int_{\mathbb{R}} |(\mathcal{M}(u))(x)|^p\, w(x)\,\ud x \;\leqslant\; \int_{\mathbb{R}} |u(x)|^p \,w(x) \,\ud x\,.
\end{equation}
If, for some function $L(x)$ one has $|L(x)|\leqslant A(x)$ in 
$\mathbb{R}$ for some $A\in L^1(\mathbb{R})$ which is 
bounded, non-negative, even, and non-increasing 
on $(0,+\infty)$, then 
$|(L * u)(x)|\leqslant C(\mathcal{M}(u))(x)$, hence the convolution 
operator on $\mathbb{R}$ with 
the function $L(x)$ is bounded in $L^p(\mathbb{R}, w(x)dx)$.
%
 \item[(iv)] The function $|x|^a$ is an $A_p$-weight on 
$\mathbb{R}$ if and only if $a\in(-1,p-1)$.
\end{itemize}
\end{thm}

Concerning part (i) we refer to 
\cite[Remark 4.1.1]{Grafakos_ModernFourier}. 
Part (ii) is a corollary of 
\cite[Theorem 7.4.6]{Grafakos_ClassialFourier}.
The first and second statement of part (iii) 
are respectively \cite[Theorem 1, Section V.3]{Stein_HA_1993} and the Proposition in page 57 
of \cite[Section II.2.1]{Stein_HA_1993}.
For part (iv) we refer to 
\cite[Example 7.1.7]{Grafakos_ClassialFourier}.

\section{Stationary representation of wave operators}
\label{sec:Stationary_Repr_W_ops}

Following a standard procedure \cite{Kuroda1978_intro_scatt_theory}, 
in order to prove the $L^p$-boundedness of $W^{+}_{\a, Y}$ 
we want to represent $W^{+}_{\a, Y}$ by means of the boundary 
values attained by the resolvents of $H_{\alpha,Y}$ and $H_0$ 
on the reals.

To this aim, we introduce the operators $\Omega_{jk}$, $j,k\in\{1,\dots,N\}$, acting on $L^2(\mathbb{R}^3))$, defined by
\begin{equation}\label{eq:def_of_Omega_jk}
\begin{split}
&(\Omega_{jk}f)(x)\;:=\; 
\lim_{\delta\downarrow 0}\frac{1}{\pi\ii}
\int_0^{+\infty}\!\!\ud\lambda\,\lambda\,e^{-\delta\lambda}\, \\
& \qquad \times \left(
\int_{\mathbb{R}^3} 
(\Gamma_{\alpha,Y}(-\lambda)^{-1})_{jk}\,\mathcal{G}_{-\lambda}(x)
\big(\mathcal{G}_{\lambda}(y)- \mathcal{G}_{-\lambda}(y)\big)\,
f(y)\ud y\,\right)\,,
\end{split}
\end{equation}
and we also introduce the translation operators $T_{x_0}:L^2(\mathbb{R}^3)\to L^2(\mathbb{R}^3)$, $x_0\in\mathbb{R}^3$, defined by
\begin{equation}
(T_{x_0}f)(x)\;:=\;f(x-x_0)\,.
\end{equation}

First of all, we show that the $\Omega_{ij}$'s are well-defined. It is convenient to re-write $\Omega_{jk}$ by using the spherical mean $M_u$ of a given function $u$, namely
\begin{equation}\label{eq:spherical_mean}
M_u(r)\;:=\; \frac1{4\pi}
\int_{{\mathbb S}^2}u(r\w)\,\ud\omega \,, \qquad r\in\mathbb{R}\,.
\end{equation} 
Observe that $\mathbb{R}\ni r\mapsto M_u(r)$ is even. 
It is also convenient to define the matrix-valued function $\lambda\mapsto F(\lambda):=(F_{jk}(\lambda))_{jk}$ by 
\begin{equation}\label{eq:def_Fjk}
F_{jk}(\lambda)\;:=\;\mathbf{1}_{(0,+\infty)}(\lambda)\,
\lambda\,(\Gamma_{\alpha,Y}(-\lambda)^{-1})_{jk}\,, \quad 
j,k\in\{1,\dots,N\}\,, 
\end{equation}
where $\mathbf{1}_\Lambda$ denotes the characteristic function of the 
set $\Lambda$.  

\begin{lem}\label{lm:3-2}~
\begin{itemize}
\item[(i)] 
The function $\lambda\mapsto F(\lambda)$ 
of \eqref{eq:def_Fjk} is \emph{smooth} and \emph{uniformly bounded} 
on $\mathbb{R}$, and
\begin{equation}\label{eq:conv_Fjk_F}
\lim_{\lambda\to+\infty}F(\lambda)\;=\;-4\pi\ii\mathbbm{1}\,.
\end{equation}
\item[(ii)] The limit \eqref{eq:def_of_Omega_jk} exists in 
$L^2({\mathbb R}^3)$ and $\Omega_{jk}$ may be written in the form 
\begin{equation}\label{Omega_intermediate_components}
 (\Omega_{jk}{u})(x)\;=\;\frac{1}{\,\ii(2\pi)^{\frac{3}{2}}|x|\,}
\int_{\mathbb R} \!\! e^{-\ii\lambda|x|}
\,F_{jk}(\lambda) \widehat{(rM_u)}(-\lambda)\ \ud\lambda\,.
\end{equation}
\end{itemize}
\end{lem}

If we introduce the distributional Fourier transform of 
$F_{jk}(\lambda)$ as  
\begin{equation}\label{ljk}
L_{jk}(\r)\;:=\;\frac{1}{\sqrt{2\pi}}\lim_{\delta \downarrow 0} 
\int_{0}^{+\infty}\!\!\ud\lambda\,  
e^{-\delta\lambda} e^{-\ii\lambda \rho} F_{jk}(\lambda)\, ,  
\end{equation}
it follows from \eqref{Omega_intermediate_components} that 
\begin{equation}\label{Omega_intermediate-a_components}
 (\Omega_{jk}{u})(x)\;=\;\frac{1}{\,\ii(2\pi)^{\frac{3}{2}}|x|\,}\,
(L_{jk} \ast rM_u)(|x|)\,.
\end{equation}

\begin{proof}[Proof of Lemma \ref{lm:3-2}] (i) Recall from Theorem \ref{thm:general_properties}(iv) 
that $\mathbb{C}\ni z\mapsto\Gamma_{\alpha,Y}^{-1}(z)$ is a 
meromorphic function, whose poles in the complex upper half-plane are 
all located on the positive imaginary axis. In particular, the only 
pole on the real line can be $z=0$, in which case it is a pole of 
order one, owing to Lemma \ref{lem:pole_of_order_1}. This implies 
that $\lambda\mapsto \lambda \,\Gamma_{\alpha,Y}^{-1}(-\lambda)$ is 
smooth and bounded on compact sets of $(0,+\infty)$, and so is 
$\lambda\mapsto F(\lambda)$ on compact sets of $\mathbb{R}$. 
Concerning the behaviour as $\lambda\to +\infty$, we see from 
\eqref{ga-def} that
\[
\Gamma_{\alpha,Y}(-\lambda)\;=\;-(4\pi\ii)^{-1}\lambda\mathbbm{1}
+R(\lambda)
\]
for some symmetric matrix $R(\lambda)$ that is uniformly bounded 
for $\lambda\in (0,\infty)$. Thus, as $\lambda\to +\infty$, 
\[
\frac{\Gamma_{\alpha,Y}(-\lambda)}{\lambda}\;=\;-(4\pi\ii)^{-1}\mathbbm{1}+\frac{R(\lambda)}{\lambda}\; \to \;-(4\pi\ii)^{-1}\mathbbm{1}\,,
\]
which proves \eqref{eq:conv_Fjk_F}.

(ii) Let $u\in C^\infty_0(\mathbb{R}^3)$. Then, for 
$\lambda\in\mathbb{R}$,
\begin{equation*}
\int_{\mathbb{R}^3} \mathcal{G}_{\lambda}(y) u(y)\,\ud y \;=\;\int_{\mathbb{R}^3}\frac{e^{\ii\lambda|y|}}{\,4\pi |y|\,}u(y) \,\ud y 
\;=\;\int_0^{+\infty} \!\!e^{\ii\lambda r} rM_u(r) \,\ud r\,. 
\end{equation*}
Since $\mathbb{R}\ni r\mapsto M_u(r)$ is even, the identity 
above yields
\begin{equation}\label{eq:rMu}
\begin{split}
\int_{\mathbb{R}^3}\big(\mathcal{G}_{\lambda}(y)-\mathcal{G}_{-\lambda}(y)\big)u(y)\,
\ud y
\;&=\; \int_{\mathbb{R}} e^{i\lambda r} \,rM_u (r)\,\ud r \\
&=\;\sqrt{2\pi}\,\widehat{(rM_u)}(-\lambda)\, 
\end{split}
\end{equation}
and \eqref{eq:def_of_Omega_jk} may be rewritten as 
\begin{equation}\label{eq:def_of_Omega_jk-a}
(\Omega_{jk}u)(x)\;=\; 
\lim_{\delta\downarrow 0}\frac{1}{(2\pi)^{\frac32}\ii |x|}
\int_0^{+\infty}\!\! e^{-\delta\lambda}\, 
F_{jk}(\lambda) e^{-i\lambda|x|}\,\widehat{(rM_u)}(-\lambda) 
\ud\lambda \, . 
\end{equation}
Here $\widehat{(rM_u)}(-\lambda)$ is a square integrable function of 
$\lambda \in {\mathbb R}$ because Parseval's identity and H\"older's 
inequality yield  
\[
\|(\widehat{r M_u})(-\lambda)\|_{L^2(\mathbb R)}
= \|r M_u\|_{L^2(\mathbb R)} 
\leqslant (\sqrt{\pi})^{-1}\|u\|_{L^2({\mathbb R}^3)}.
\]
Since $F_{jk}(\lambda)$ is bounded, the Fourier inversion formula 
implies that the limit $\delta\downarrow 0$ in 
\eqref{eq:def_of_Omega_jk-a} exists in $L^2(\mathbb{R}^3_x)$ 
and \eqref{Omega_intermediate_components} follows. 
\end{proof} 

The main result of this Section is the following representation formula for the wave operator.

\begin{prop}~
Let $u,v \in L^2(\mathbb{R}^3)$. Then, 
\begin{equation} \label{eq:stationary_rep_w}
\langle W^{+}_{\a, Y} u,v\rangle\;=\;\langle u,v\rangle 
+\sum_{j,k=1}^N \langle T_{y_j}
\Omega_{jk}T_{y_k}^{\,*}u,v\rangle\,.
\end{equation}
\end{prop}

\begin{proof} 
It suffices to prove \eqref{eq:stationary_rep_w} for $u, v \in C_0^\infty(\mathbb{R}^3)$. 

The limit \eqref{eq:main_thm_W_op} when $t\to +\infty$ 
equals its Abel limit, thus we re-write
\begin{equation}\label{eq:after_Abel_limit}
\langle W^{+}_{\a, Y} u,v\rangle\;=\; 
\lim_{\ep\downarrow 0}\,2\ep \int_0^{+\infty} 
\langle e^{-\ii t (H_0-\ii\ep{\mathbbm{1}})}u, 
e^{-\ii t(H_{\alpha,Y}-\ii\ep{\mathbbm{1}})}v\rangle \,\ud t\,. 
\end{equation}
Let now $\mu\in\mathbb{R}$. Exploiting the Fourier transform
\[
(H_0-(\mu+\ii\varepsilon){\mathbbm{1}})^{-1}\;
=\;\ii\int_0^{+\infty}e^{\ii\mu t}\,
e^{-\ii t (H_0-\ii\varepsilon{\mathbbm{1}})}\,\ud t\qquad (\varepsilon>0)
\]
(and the analogue for $H_{\alpha,Y}$), Parseval's formula in 
the r.h.s.~of \eqref{eq:after_Abel_limit} yields  
\begin{equation} \label{abel}
\langle W^{+}_{\a, Y} u,v\rangle\;=\; 
\lim_{\ep\downarrow 0}\frac{\ep}{\pi} \int_{\mathbb{R}}
\langle R_0(\lambda+\ii\ep)u, R_{\alpha,Y}(\lambda+\ii\ep)v\rangle\,\ud\lambda\,. 
\end{equation}
Here and henceforth
\begin{equation}
\begin{array}{rlll}
  R_0(\mu)\!\!&:=\;(H_0-\mu\mathbbm{1})^{-1} & & \mu\in\mathbb{C}\setminus\![0,+\infty)\,, \\
  R_{\alpha,Y}(\mu)\!\!&:=\;(H_{\alpha,Y}-\mu\mathbbm{1})^{-1} & & \mu\in\mathbb{C}\setminus\sigma(H_{\alpha,Y}) \,,
\end{array}
\end{equation}
that is, the resolvents of the operators $H_0$ and $H_{\alpha,Y}$.

Substituting $R_{\alpha,Y}(\lambda+\ii\ep)$ in the r.h.s.~of 
\eqref{abel}  with the resolvent identity 
\eqref{eq:resolvent_identity}, one obtains
\begin{equation}\label{eq:after_replacing_res_id}
\begin{split}
\langle W^{+}_{\a, Y} u,v\rangle\;&=\;  
\lim_{\ep\downarrow 0}\frac{\ep}{\pi} 
\int_{\mathbb{R}}\langle R_0(\lambda+\ii\ep)u, R_0(\lambda+\ii\ep)v\rangle\,\ud\lambda \\
& \qquad+ \;\lim_{\ep\downarrow 0}\frac{\ep}{\pi} 
\sum_{j,k=1}^N
\int_{\mathbb{R}} (\Gamma_{\alpha,Y}(\sqrt{\lambda+\ii\ep})^{-1})_{jk} \,\times \\
&\qquad\qquad\times\,\langle \,R_0(\lambda+\ii\ep)u\,,\,
\mathcal{G}_{\sqrt{\lambda+\ii\ep}}^{y_j}\otimes   
\overline{\mathcal{G}_{\sqrt{\lambda+\ii\ep}}^{y_k}}\;v\,\rangle\,\ud\lambda\,.
\end{split}
\end{equation}
The first summand in the r.h.s.~of 
\eqref{eq:after_replacing_res_id} gives
\[
\begin{split}
\frac{\ep}{\pi} 
\int_{\mathbb{R}}\langle R_0(\lambda+\ii\ep)&u, R_0(\lambda+\ii\ep)v\rangle\,\ud\lambda \;=\;\frac{\ep}{\pi} 
\int_{\mathbb{R}}\langle u, R_0(\overline{\lambda+\ii\ep})R_0(\lambda+\ii\ep)v\rangle\,\ud\lambda  \\
&=\;\frac{\ep}{\pi}\int_{\mathbb{R}}\ud\lambda\int_{\sigma(H_0)}\langle u,E^{(H_0)}(\ud h) v\rangle\,\frac{1}{\,(h-\lambda)^2+\varepsilon^2\,} \\
&=\;\int_{\sigma(H_0)}\langle u,E^{(H_0)}(\ud h)v\rangle\;\frac{1}{\pi}\int_{\mathbb{R}}\ud\lambda\,\frac{\varepsilon}{\,(h-\lambda)^2+\varepsilon^2\,}\;=\;\langle u,v\rangle\,,
\end{split}
\]
thus \eqref{eq:after_replacing_res_id} reads 
\begin{equation}\label{eq:compu-1}
\begin{split}
\langle W^{+}_{\a, Y} u,v\rangle\;=\;\langle u,v\rangle +\; &\lim_{\ep\downarrow 0}\frac{\ep}{\pi} \sum_{j,k=1}^N
\int_{\mathbb{R}} (\Gamma_{\alpha,Y}(\sqrt{\lambda+\ii\ep})^{-1})_{jk}\,\times \\
&\qquad\times\,\langle u,R_0(\overline{\lambda+\ii\ep})\,\mathcal{G}_{\sqrt{\lambda+\ii\ep}}^{y_j}\rangle\,\langle  
\overline{\mathcal{G}_{\sqrt{\lambda+\ii\ep}}^{y_k}},v\,\rangle\,\ud\lambda\,.
\end{split}
\end{equation}

We recall that 
$\sqrt{z}$ is chosen in the upper complex half plane 
and, for $z\in \mathbb{C} \setminus [0,\infty)$, 
\begin{equation} \label{intre}
\!\!\!\!\!\!\!\!\!\mathcal{G}_{\sqrt{z}}^{y_j}(x)\;=\;\frac1{(2\pi)^3}\int_{\mathbb{R}^3}\frac{\:e^{\ii p(x-y_j)}}{p^2-z}\, \ud p\; 
\left(\equiv \lim_{L\to \infty}\frac1{(2\pi)^3}
\int_{|p|<L}\!\!\frac{\:e^{\ii p(x-y_j)}}{p^2-z} \, \ud p \right). \!\!\!\!
\end{equation} 
Thus, for $z\equiv\lambda+\ii\varepsilon$, both $\sqrt{\lambda+\ii\ep}$ and $\sqrt{\lambda-\ii\ep}$ belong to $\mathbb{C}^{+}$, and we compute
\begin{equation}\label{real}
\begin{split}
\!\!\!\frac{\ep}{\pi}\, R_0&(\lambda-\ii\ep)\,\mathcal{G}_{\sqrt{\lambda+\ii\ep}}^{y_j}(x)\;=\;\frac1{(2\pi)^3}\,\frac{\ep}{\pi}\int_{\mathbb{R}^3}\frac{e^{\ii p(x-y_j)}}
{(p^2-\lambda+\ii\ep)(p^2-\lambda-\ii\ep)}\,\ud p \\
&=\;\frac1{(2\pi)^3}\,\frac1{2\pi \ii}\int_{\mathbb{R}^3}{e^{\ii p(x-y_j)}}
\Big(\frac1{(p^2-\lambda-\ii\ep)}-\frac1{(p^2-\lambda+\ii\ep)}\Big) \, \ud p \\
&=\;\frac1{2\pi \ii}
\Big(\mathcal{G}_{\sqrt{\lambda+\ii\ep}}^{y_j}(x)- 
\mathcal{G}_{\sqrt{\lambda-\ii\ep}}^{y_j}(x)\Big)\,.
\end{split}
\end{equation}
The second summand in the r.h.s.~of \eqref{eq:compu-1} can be then 
written as
\begin{equation}\label{eq:compu-2}
\begin{split}
\lim_{\ep\downarrow 0}\sum_{j,k=1}^N\frac{1}{2\pi\ii}
\int_{\mathbb{R}}\ud\lambda&\, \Big(\int_{\mathbb{R}^3}\ud y\,\overline{u(y)}\,\Big(\mathcal{G}_{\sqrt{\lambda+\ii\ep}}^{y_j}(y)- 
\mathcal{G}_{\sqrt{\lambda-\ii\ep}}^{y_j}(y)\Big)  \\
&\times\,(\Gamma_{\alpha,Y}(\sqrt{\lambda+\ii\ep})^{-1})_{jk}\,\Big(\int_{\mathbb{R}^3}\ud x \,\mathcal{G}_{\sqrt{\lambda+\ii\ep}}^{y_k}(x)\,v(x)\Big)\,.
\end{split}
\end{equation}

Because $u$ and $v$ are smooth and with compact support, 
an integration by parts shows that both the $\ud x$-integral 
and the $\ud y$-integral in \eqref{eq:compu-2} above are 
bounded by $C\langle\lambda\rangle^{-\frac{1}{2}}$ \emph{uniformly in} 
$\varepsilon$. Moreover, as established in 
Lemma \ref{lem:pole_of_order_1}, the matrix 
$\Gamma_{\alpha,Y}(\sqrt{\lambda+\ii\ep})^{-1}$ has the singularity 
$(\sqrt{\lambda+\ii\varepsilon})^{-1}$ near $\lambda=0$ 
(in the limit $\varepsilon\downarrow 0$) if $H_{\alpha,Y}$ has 
a zero-energy resonance, whereas it is bounded otherwise, 
with $\|\Gamma_{\alpha,Y}(\sqrt{\lambda+\ii\ep})^{-1}\|
\leqslant\;C\langle\lambda\rangle^{-\frac{1}{2}}$. 
Therefore the $\lambda$-integrand is uniformly bounded by 
$C\lambda^{-\frac{1}{2}}\langle\lambda\rangle^{-1}$, dominated convergence is 
applicable in \eqref{eq:compu-2} above, the 
$\ud\lambda$-integration and the 
$\varepsilon\downarrow 0$-limit can be exchanged, 
and \eqref{eq:compu-2} becomes  
\begin{equation}\label{eq:compu-3}
\begin{split}
\sum_{j,k=1}^N\frac{1}{2\pi\ii}
\int_{\mathbb{R}}\ud\lambda&\, \Big(\int_{\mathbb{R}^3}\ud y\,\overline{u(y)}\,\Big(\mathcal{G}_{\sqrt{\lambda+\ii0}}^{y_j}(y)- 
\mathcal{G}_{\sqrt{\lambda-\ii0}}^{y_j}(y)\Big) \\
&\times\,(\Gamma_{\alpha,Y}(\sqrt{\lambda+\ii0})^{-1})_{jk}\,\Big(\int_{\mathbb{R}^3}\ud x \,\mathcal{G}_{\sqrt{\lambda+\ii0}}^{y_k}(x)\,v(x)\Big)\,.
\end{split}
\end{equation}

Owing to the difference $\mathcal{G}_{\sqrt{\lambda+\ii0}}^{y_j}- 
\mathcal{G}_{\sqrt{\lambda-\ii0}}^{y_j}$, we see that the $\lambda$-integration in \eqref{eq:compu-3} is only effective when $\lambda\geqslant 0$. Indeed, if $\lambda<0$, then $\sqrt{\lambda\pm\ii 0}=\ii\sqrt{|\lambda|}$ and the integrand vanishes.
We then consider \eqref{eq:compu-3} only with $\lambda\in[0,+\infty)$ and with the change of variable $\lambda\mapsto\lambda^2$ we obtain
\begin{equation}\label{eq:compu-4}
\begin{split}
&\textrm{second summand in the r.h.s.~of \eqref{eq:compu-1}}\;= \\
&=\;\sum_{j,k=1}^N\frac{1}{\pi\ii}
\int_0^{+\infty}\!\!\ud\lambda\,\lambda\, 
\Big(\int_{\mathbb{R}^3}\ud y\,\overline{u(y)}\,
\Big(\mathcal{G}_{\lambda}^{y_j}(y)- 
\mathcal{G}_{-\lambda}^{y_j}(y)\Big) \\
&\qquad\qquad\qquad\times\,
(\Gamma_{\alpha,Y}(\lambda)^{-1})_{jk}\,
\Big(\int_{\mathbb{R}^3}\ud x \,
\mathcal{G}_{\lambda}^{y_k}(x)\,v(x)\Big) \\
&=\;\lim_{\delta\downarrow 0}\sum_{j,k=1}^N\frac{1}{\pi\ii}
\int_0^{+\infty}\!\!\ud\lambda\,\lambda\, e^{-\delta\lambda}\,
\Big(\int_{\mathbb{R}^3}\ud y\,\overline{u(y+y_k)}\,
\big(\mathcal{G}_{\lambda}(y)- 
\mathcal{G}_{-\lambda}(y)\big) \\
&\qquad\qquad\qquad\times\,(\Gamma_{\alpha,Y}(\lambda)^{-1})_{jk}\,
\Big(\int_{\mathbb{R}^3}\ud x \,
\mathcal{G}_{\lambda}^{y_j}(x)\,v(x)\Big) \\
&=\;\lim_{\delta\downarrow 0}\int_{\mathbb{R}^3}\ud x\,v(x)
\sum_{j,k=1}^N\,\int_0^{+\infty}\!\!\ud\lambda\,\lambda\, e^{-\delta\lambda}\!\int_{\mathbb{R}^3}\ud y \\
&\qquad\times\Big(\overline{
\frac{1}{\pi\ii}\,
(\Gamma_{\alpha,Y}(-\lambda)^{-1})_{jk}\,
\mathcal{G}_{-\lambda}^{y_j}(x)\,\big(\mathcal{G}_{\lambda}(y)- 
\mathcal{G}_{-\lambda}(y)\big)u(y+y_k)}\Big).
\end{split}
\end{equation}
In the first step of \eqref{eq:compu-4} above we used the fact that $\sqrt{\lambda^2\pm i0}= \pm \lambda $ for $\lambda>0$. In the second step, the insertion of the exponential cut-off $e^{-\delta\lambda}$ is justified by the fact that the $\lambda$-integrand is uniformly bounded by $C\langle\lambda\rangle^{-\frac{5}{2}}$, as discussed above; we also exchanged $j\leftrightarrow k$, using the fact that $\Gamma_{\alpha,Y}(\lambda)^{-1}$ is symmetric, and made the change of variable $y\mapsto y+y_k$, using \eqref{eq:def_of_the_Gs}. In the third step we used the properties $\overline{\mathcal{G}_\lambda(x)}=\mathcal{G}_{-\lambda}(x)$ and $\overline{\Gamma_{\alpha,Y}(\lambda)^{-1}}=\Gamma_{\alpha,Y}(-\lambda)^{-1}$ that follow, respectively, from \eqref{eq:def_of_the_Gs} and \eqref{ga-def}.
The identity \eqref{eq:stationary_rep_w} then follows immediately from 
\eqref{eq:compu-4}.  
\end{proof}

Summarising so far, we produced the representation 
\eqref{eq:def_of_Omega_jk}-\eqref{eq:stationary_rep_w} 
of the kernel of the wave operator $W^{+}_{\a, Y}$.
Because of the obvious $L^p$-boundedness of $T_{x_0}$, in order to prove Theorem \ref{ref:main_thm} 
it suffices to study the $L^p$-boundedness or 
unboundedness of each $\W_{jk}$, that is, to consider the quantities 
\begin{equation}\label{forthe}
\|\W_{jk}u\|_{L^p(\mathbb{R}^3)}^p \;=\; \frac{4\pi}{(2\pi)^{3p/2}}
\int_0^{+\infty} |(L_{jk}\ast \rho M_u)(\r)|^p \,\r^{2-p}\,\ud\r\,,
\end{equation}
whose expression follows from \eqref{Omega_intermediate-a_components}.

For a more compact notation, it is convenient to introduce the matrix functions 
\begin{equation}\label{eq:def_F_matrix}
L(\r)\;:=\; \left(L_{jk}(\r)\right)_{jk}\,, 
\quad \Omega(\r)\;:=\; \left(\Omega_{jk}(\r)\right)_{jk}\,,
\end{equation}
in terms of which
\begin{equation}\label{Omega_intermediate}
 (\Omega{u})(x)\;=\;\frac{1}{\,\ii(2\pi)^{\frac{3}{2}}|x|\,}
\int_0^{+\infty}\!\! e^{-\ii\lambda|x|}
\,F(\lambda) \widehat{(rM_u)}(-\lambda)\ \ud\lambda
\end{equation}
and
\begin{equation}\label{Omega_intermediate-a}
 (\Omega{u})(x)\;=\;\frac{1}{\,\ii(2\pi)^{\frac{3}{2}}|x|\,}\,
(L \ast rM_u)(|x|)\,.
\end{equation}
The additional formulas 
\eqref{Omega_intermediate}/\eqref{Omega_intermediate-a} have the 
virtue of reducing the problem to the estimate of singular integral 
operators in one dimensions and will play an important role in our 
next arguments -- although in certain steps we need to go back to 
the more complicated, but more flexible expression 
\eqref{eq:def_of_Omega_jk}.


\section{$L^p$-bounds for the single centre case.}\label{sec:Lp-bounds}

In this Section and in the two following ones we present the proof of Theorem \ref{ref:main_thm}. 
In fact, only the statements concerning the boundedness and the unboundedness  of $W^{\pm}_{\a, Y}$ need be proved, because the existence of $W^{\pm}_{\a, Y}$ in $L^2(\mathbb{R}^3)$ and their completeness follow at once from the Birman-Kato-Pearson Theorem \cite{rs3}, due to the fact (Theorem \ref{thm:general_properties}(i), identity \eqref{eq:resolvent_identity}) that the resolvent difference $R_{\alpha,Y}(z)-R_0(z)$ is a rank-$N$ operator. 

We also observe that, by virtue of Lemma \ref{reality}, the complex 
conjugation $u\mapsto \mathcal{C}u:=\overline{u}$
reverses the direction of time, i.e., 
\begin{equation}
\mathcal{C}^{-1}e^{-\ii t H_{\alpha,Y}}\mathcal{C} \;=\; e^{\ii t H_{\alpha,Y}}, \qquad 
\mathcal{C}^{-1}e^{-\ii t H_0}\mathcal{C} \;=\; e^{\ii t H_0}\,,
\end{equation}
whence
\begin{equation}\label{eq:W-_W+_charge}
W^{-}_{\a, Y} \;=\;\mathcal{C}^{-1}\, W^{+}_{\a, Y}\,\mathcal{C}\,.
\end{equation} 
Thus, once the $L^p$-boundedness is proved for $W^{+}_{\a, Y}$ and 
all $p\in(1,3)$, the same result follows for $W^{-}_{\a, Y}$ via 
\eqref{eq:W-_W+_charge}. Analogously, it suffices to prove the $L^p$-unboundedness of $W^{+}_{\a, Y}$, for $p=1$ and $p\in[3,\infty)$, in order to have same result for $W^{-}_{\a, Y}$.

We start with the proof of the boundedness part of Theorem \ref{ref:main_thm} in the special case of $N=1$ centre. This case is simpler, for the oscillating terms $\mathcal{G}_\lambda^{y_j,y_k}$ are now absent, nevertheless it 
retains most of the essential ideas needed in the proof of the 
general case, which is the object of the following Section \ref{sec:Lp-bounds_N}.

We shall control the two regimes $p\in(1,\frac{3}{2})$ and $p\in(\frac{3}{2},3)$ separately. Then the overall $L^p$-boundedness for $p\in(1,3)$ follows by interpolation.

\subsection{$L^p$-boundedness of $W^{+}_{\a, Y}$ for $N=1$ and $p\in(\frac{3}{2},3)$}

In this regime the proof is based on Theorem \ref{thm:harmonic_analysis_results} and on the following fact.

\begin{lem}\label{lem:kernel-0} Suppose that $[0,+\infty)\ni\lambda\mapsto W(\lambda)$ is a smooth and bounded function such that $\lambda\mapsto W'(\lambda)$ and $\lambda\mapsto\lambda W''(\lambda)$ are both integrable.  
Let $Z(\r)$, $\r\in \mathbb{R}$, be the Fourier transform of $W(\lambda)$, 
in the sense of distributions, defined by 
\[
Z(\r)\;=\; \frac1{\sqrt{2\pi}}
\int_0^{+\infty} \!\ud\lambda\,e^{-\ii\lambda\r} \,W(\lambda)\, .
\]
Then,
the convolution operator with $Z(\r)$ is a Calder\'on-Zygmund operator on $\mathbb{R}$.
In particular, the operator $u \mapsto L\ast u$, where $L$ is defined in \eqref{ljk} for the case $N=1$, 
is of Calder\'on-Zygmund type. 
\end{lem}

\begin{proof}
 The operator of convolution with $Z$ is bounded in $L^2(\mathbb{R})$ because $Z$ is the Fourier 
transform of a bounded function $W$. Integration by parts, using 
$e^{-\ii\lambda\r}=\ii\r^{-1}\partial_\lambda e^{-\ii\lambda\r}$, yields    
\begin{equation*}  
Z(\r)\;=\; \frac{\ii}{\rho\sqrt{2\pi}\,}W(0)- 
\frac{\ii}{\rho\sqrt{2\pi}\,}
\int_{0}^{+\infty}\!\!\ud\lambda\, e^{-\ii\lambda\r}\,W'(\lambda) 
\;\leqslant_{|\,\cdot\,|}\; \frac{C}{|\r|}\, , \quad \r\not=0\,,
\end{equation*}
and differentiating further in $\r$ yields 
\begin{equation*}
\begin{split}
 Z'(\r)\;&=\; -\frac{\,\ii\,W(0)}{\rho^2\sqrt{2\pi}\,}+\frac{\ii}{\rho^2\sqrt{2\pi}}\int_{0}^{+\infty}\!\!\ud\lambda\, e^{-\ii\lambda\r}\,W'(\lambda)\\
 &\qquad -\frac{1}{\rho\sqrt{2\pi}}\int_{0}^{+\infty}\!\!\ud\lambda\, e^{-\ii\lambda\r}\,\lambda\,W'(\lambda)\,.
\end{split}
\end{equation*}
The first two summands in the r.h.s.~above are obviously bounded in absolute value by $C|\r|^{-2}$ 
for $\r\not=0$; so too is the third summand, as follows from integration by parts:
\[
\begin{split}
 \frac{-1}{\rho\sqrt{2\pi}}\int_{0}^{+\infty}\!\!\ud\lambda\, e^{-\ii\lambda\r}\,\lambda\,W'(\lambda)\;&=\;\frac{\ii}{\rho^2\sqrt{2\pi}\,}\int_{0}^{+\infty}\!\!\ud\lambda\, e^{-\ii\lambda\r} 
(W'(\lambda)+ \lambda W''(\lambda)) \\
&\leqslant_{|\,\cdot\,|}\; \frac{C}{\r^2}\, , \quad \r\not=0\,.
\end{split}
\]
Thus, we conclude from Theorem \ref{thm:harmonic_analysis_results}(i) that $u\mapsto Z*u$ is a Calder\'on-Zygmund operator on $\mathbb{R}$. Concerning the second statement of the thesis, we see that in the case $N=1$ 
\eqref{eq:def_Fjk} reads
\begin{equation}\label{eq:FN1}
F(\lambda)\;=\;\lambda\Big(\alpha + \frac{\ii\lambda}{4\pi}\Big)^{-1}\,.
\end{equation}
$F$ is therefore bounded and smooth on $[0,+\infty)$  and  both $F'(\lambda)$ and $\lambda\, F''(\lambda)$ are integrable, whence the conclusion for the operator of convolution by $L$ defined in \eqref{ljk}. 
\end{proof}

The proof of the $L^p$-boundedness of $W^{+}_{\a, Y}$ for $N=1$ and $p\in(\frac{3}{2},3)$ then becomes particularly simple. First, we recall from 
\eqref{forthe} that
\begin{equation*}
\|\W u\|_{L^p(\mathbb{R}^3)}^p \;=\; \frac{4\pi}{(2\pi)^{3p/2}}
\int_0^{+\infty} |(L\ast \rho M_u)(\r)|^p \,\r^{2-p}\,\ud\r\,,
\end{equation*}
where $\rho^{2-p}$ is an $A_p$-weight for $p\in(\frac{3}{2},3)$ (Theorem \ref{thm:harmonic_analysis_results}(iv)) and the convolution with $L$ is a Calder\'on-Zygmund operator on $\mathbb{R}$ (Lemma \ref{lem:kernel-0}). Then it follows from Theorem \ref{thm:harmonic_analysis_results}(ii) that
\begin{equation}\label{eq:Opb_final}
\begin{split}
\|\W u\|_{L^p(\mathbb{R}^3)}^p \;&\leqslant\;
\int_0^{+\infty} |(\rho M_u)(\r)|^p \,\r^{2-p}\,\ud\r\;=\;
\int_0^{+\infty} |M_u(\r)|^p \,\r^{2}\,\ud\r \\
&\leqslant\;C_p\,\|u\|_{L^p(\mathbb{R}^3)}^p
\end{split}
\end{equation}
for some constant $C_p>0$, whence the conclusion.

\subsection{$L^p$-boundedness of $W^{+}_{\a, Y}$ for $N=1$ and $p\in(1,\frac{3}{2})$}\label{subsec:N1_p132}

In the regime $p\in(1,\frac{3}{2})$ the general harmonic analysis 
treatment provided by Theorem \ref{thm:harmonic_analysis_results} 
only allows us to find an $L^p$-bound to \emph{part} 
of the function (see \eqref{Omega_intermediate_components} above)
\begin{equation*}
 (\Omega{u})(x)\;=\;\frac{1}{\,\ii(2\pi)^{\frac{3}{2}}|x|\,}
\int_0^{+\infty}\!\! e^{-\ii\lambda|x|}
\,F(\lambda) \widehat{(rM_u)}(-\lambda)\ \ud\lambda\,,
\end{equation*}
whereas for the remaining part we need to produce further analysis. 

Integrating by parts the above expression of $\Omega u$, using 
$e^{-\ii\lambda\r}=\ii\r^{-1}\partial_\lambda e^{-\ii\lambda\r}$, yields    
\begin{equation}\label{eq:O=O1+O2}
 \Omega u \;=\;\Omega_1 u+\Omega_2 u\,,
\end{equation}
where
\begin{equation}\label{eq:defO1O2}
 \begin{split}
  (\Omega_1 u)(x)\;&:=\;\frac{-\ii}{(2\pi)^{\frac{3}{2}}|x|^2 }
\int_0^{+\infty}\!\! e^{-\ii\lambda |x| }
\,F(\lambda) \widehat{(r^2 M_u)}(-\lambda)\ \ud\lambda\,, \\
 (\Omega_2 u)(x)\;&:=\;\frac{-1}{(2\pi)^{\frac{3}{2}}|x|^2 }
\int_0^{+\infty}\!\! e^{-\ii\lambda|x| }
\,F'(\lambda) \widehat{(r M_u)}(-\lambda)\ \ud\lambda\,.
 \end{split}
\end{equation}

Now, concerning $\Omega_1 u$, we re-write
\begin{equation}\label{eq:O1rewritten}
 (\Omega_1 u)(x)\;=\;\frac{-\ii}{(2\pi)^{\frac{3}{2}}|x|^2 } 
(L\ast r^2 M_u)(|x|)
\end{equation}
with $L$ given by \eqref{ljk}. Owing to Lemma \ref{lem:kernel-0}, $u \mapsto L\ast u$ is a Calder\'on-Zygmund operator, and owing to Theorem \ref{thm:harmonic_analysis_results}(iv), $|x|^{2-2p}$ is an $A_p$-weight on $\mathbb{R}$ for $p\in(1,\frac{3}{2})$. Therefore, 
\begin{equation}\label{eq:bound_on_O1}
 \begin{split}
  \|\Omega_1 u\|_{L^p(\mathbb{R}^3)}^p\;&=\;\frac{4\pi}{(2\pi)^{3p/2}}
\int_0^{+\infty} |(L\ast \rho^2 M_u)(\r)|^p \,\r^{2-2p}\,\ud\r \\
&\leqslant\;\int_0^{+\infty} |\rho^2 M_u(\r)|^p \,\r^{2-2p}\,\ud\r\;=\;\int_0^{+\infty} | M_u(\r)|^p \,\r^{2}\,\ud\r \\
&\leqslant\;C_p\,\| u\|_{L^p(\mathbb{R}^3)}^p
 \end{split}
\end{equation}
for some constant $C_p>0$,
where in the second step we applied Theorem \ref{thm:harmonic_analysis_results}(ii). This proves the $L^p$-boundedness of $\Omega_1$.

Concerning $\Omega_2 u$, instead, we re-write 
\begin{equation}\label{eq:O2_rewritten} 
(\Omega_2 u)(x)\;=\;\frac{-1}{(2\pi)^{\frac{3}{2}}\r^2 }
(\mathcal{L} * r M_u)(\r)\,,   
\end{equation}
where $\mathcal{L}$ is the Fourier transform of the function ${\bf 1}_{(0,\infty)}F'(\lambda)$, and
\begin{equation}
F'(\lambda)\;=\; \alpha\Big(\alpha+\frac{\ii\,\lambda}{4\pi}\Big)^{-2}\,.
\end{equation}
Thus, in the non-trivial case $\alpha\neq0$ $F'$ is smooth and bounded, and correspondingly both $F''$ and $\lambda F'''$ are integrable. This implies, through Lemma \ref{lem:kernel-0}, that $u\mapsto \mathcal{L}*u$ is a Calder\'on-Zygmund operator on $\mathbb{R}$. Since $|x|^{2-2p}$ is an $A_p$-weight on $\mathbb{R}$ for $p\in(1,\frac{3}{2})$ (Theorem \ref{thm:harmonic_analysis_results}(iv)), then Theorem \ref{thm:harmonic_analysis_results}(ii) yields
\begin{equation}\label{eq:O2partial}
\begin{split}
 \|\Omega_2 u\|_{L^p(\mathbb{R}^3)}^p\;&=\;\frac{4\pi}{(2\pi)^{3p/2}}
\int_0^{+\infty} |(\mathcal{L}* \rho\, M_u)(\r)|^p \,\r^{2-2p}\,\ud\r \\
&\leqslant\;\int_0^{+\infty} |\rho\, M_u(\r)|^p \,\r^{2-2p}\,\ud\r\;\leqslant\;C\!\int_{\mathbb{R}^3}\frac{\;|u(x)|^p}{|x|^p}\,\ud x
\end{split}
\end{equation}
for some constant $C>0$. This shows that
\begin{equation}\label{eq:first_bound_on_O2}
 \|\Omega_2 \mathbf{1}_{\{|x|\geqslant 1\}}u\|^p_{L^p(\mathbb{R}^3)}\;\leqslant\;C\,\|\mathbf{1}_{\{|x|\geqslant 1\}}u\|^p_{L^p(\mathbb{R}^3)}\,.
\end{equation}
For $L^p$-functions supported on $|x|\leqslant 1$ a further 
argument is needed. In other words, so far from \eqref{eq:bound_on_O1} 
and \eqref{eq:first_bound_on_O2} we have
\begin{equation}\label{eq:O2partial_bis}
\begin{split}
 &\|\Omega u\|^p_{L^p(\mathbb{R}^3)} \\
 &\quad\leqslant\;2\|\Omega_1 u\|^p_{L^p(\mathbb{R}^3)}+2\|\Omega_2 \mathbf{1}_{\{|x|\geqslant 1\}}u\|^p_{L^p(\mathbb{R}^3)}+2\|\Omega_2 \mathbf{1}_{\{|x|\leqslant 1\}}u\|^p_{L^p(\mathbb{R}^3)} \\
 &\quad\leqslant \;C_p\|u\|^p_{L^p(\mathbb{R}^3)}+
C\|\mathbf{1}_{\{|x|\geqslant 1\}}u\|^p_{L^p(\mathbb{R}^3)}+2\|\Omega_2 \mathbf{1}_{\{|x|\leqslant 1\}}u\|^p_{L^p(\mathbb{R}^3)}\,,
\end{split}
\end{equation}
and we are left with producing the estimate
\begin{equation}\label{eq:needed_one}
 \|\Omega_2 \mathbf{1}_{\{|x|\leqslant 1\}}u\|_{L^p(\mathbb{R}^3)}\;\leqslant\;C_p\|\mathbf{1}_{\{|x|\leqslant 1\}}u\|^p_{L^p(\mathbb{R}^3)}\,.
\end{equation}

To this aim, let us establish first the following result.

\begin{lem} \label{lem:kernel-1} 
Suppose that $[0,+\infty)\ni y\mapsto Y(y)$ is a bounded $C^1$-function such that $\lambda\mapsto\lambda^\theta Y(\lambda)$ and $\lambda\mapsto(1+ \lambda)^\theta Y'(\lambda)$ 
are both integrable for all $\theta\in (0,1)$, and let 
\begin{equation}\label{Yint}
T(x,y)\;:=\; \frac{1}{|x|^2}\int_{0}^{+\infty} 
\!\!\Big(\frac{e^{-\ii\lambda(|x|-|y|)}- e^{-\ii\lambda(|x|+|y|)}}{4\pi|y|}\Big) 
Y(\lambda) \,\ud\lambda 
\end{equation}
for $x,y\in\mathbb{R}^3$.
Then, for any $R>0$ and $p\in(1,\frac{3}{2})$, the integral operator $T$ on $\mathbb{R}^3$  
with the integral kernel 
$T(x,y)$ is 
$L^p(\Lambda_R)\to L^p(\mathbb{R}^3)$ bounded, with $\Lambda_R:=\{x\in\mathbb{R}^3\,|\,|x|\leqslant R\}$.
\end{lem}

\begin{proof} We only consider the case $R=1$, the proof for generic $R$ is similar. Let us deal with the region $|x|\leqslant 10$ first. Since 
$|e^{-\ii\lambda(|x|-|y|)}- e^{-\ii\lambda(|x|+|y|)}|\leqslant 2 (\lambda|y|)^{1-\th}$ 
for any $\theta\in (0,1)$, and since $\lambda^{1-\theta}Y\in L^1(0,+\infty)$, then
\[
|T(x,y)|\;\leqslant\;\frac{1}{\;2\pi |x|^2 |y|^{\th}\,}
\int_0^{+\infty} |Y(\lambda)|\,\lambda^{1-\theta} \,\ud\lambda \;\leqslant\; 
\frac{C_\theta}{\,|x|^2 |y|^{\th}\,}\,, \qquad |x|\leqslant 10\,,
\]
for some constant $C_\theta>0$. For fixed $p$ in $(1,\frac{3}{2})$, we take $\theta\in(0,1)$ such that
$p'\th<3$, where $p'=\frac{p}{p-1}$ as usual. With this choice, $|y|^{-\th}\in L^{p'}(\Lambda_1)$ and 
$|x|^{-2}\in L^p(\Lambda_{10})$, with $\Lambda_R=\{x\in\mathbb{R}^3\,|\,|x|\leqslant R\}$ as in the statement of the Lemma. For each $f\in L^p(\Lambda_1)$, H\"older's inequality and the above bound for $|T(x,y)|$ then imply 
\[
\|Tf\|_{L^p(\Lambda_{10})} \;
\leqslant\; C_\theta \,\||x|^{-2}\|_{L^p(\Lambda_{10})}\cdot 
\|f\|_{L^p(\Lambda_1)}\cdot \||y|^{-\th}\|_{L^{p'}(\Lambda_1)}\;=\;\kappa^-_p\,\|f\|_{L^p(\Lambda_1)}
\]
for some constant $\kappa^-_p>0$. Next, let us consider the region $|x|\geqslant 10$. Integration by parts gives
\begin{align*}
 T(x,y)\;&=\;\frac{1}{\,4\pi|x|^2|y|\,}\int_0^{+\infty}\!\!\partial_\lambda\Big(\frac{e^{-\ii\lambda(|x|-|y|)}}{-\ii(|x|-|y|)}-\frac{e^{-\ii\lambda(|x|+|y|)}}{-\ii(|x|+|y|)}\Big)\,Y(\lambda)\,\ud\lambda \\
  &=\; \frac{1}{\,4\pi\,\ii\,|x|^2|y|\,} 
\Big(\frac{1}{|x|-|y|}
- \frac{1}{|x|+|y|}\Big)\,Y(0) \tag{I}\\
  & \quad +  \frac{1}{\,4\pi\,\ii\,|x|^2|y|\,} \int_0^{+\infty}\!\!\Big(\frac{e^{-\ii\lambda(|x|-|y|)}}{|x|-|y|}-\frac{e^{-\ii\lambda(|x|+|y|)}}{|x|+|y|}\Big)\,Y'(\lambda)\,\ud\lambda\,.
    \tag{II}
\end{align*}
Since $|x|\pm |y|\geqslant \frac{9}{10}|x| \geqslant 9$ whenever $|x|\geqslant 10$ and $|y|\leqslant 1$, and since $Y$ is bounded, then clearly
\[
 |\mathrm{(I)}|\;\leqslant\;\frac{C}{\;|x|^4}\;\leqslant\;\frac{C}{\;|x|^3\,|y|^\theta}
\]
for some constant $C>0$ and any $\theta\in(0,1)$. As for the summand (II), since
\[
 \begin{split}
  \frac{e^{-\ii\lambda(|x|-|y|)}}{|x|-|y|}-\frac{e^{-\ii\lambda(|x|+|y|)}}{|x|+|y|}\;&\leqslant_{|\,\cdot\,|}\;\frac{2\,|y|}{|x|^2-|y|^2}+ 
\frac{\;(2\lambda|y|)^{1-\theta}}{|x|+|y|} \\
&\leqslant\;C\Big(\frac{|y|}{\;|x|^2}+\frac{\;(\lambda|y|)^{1-\theta}}{|x|}\Big)
 \end{split}
\]
for some constant $C>0$ and any $\theta\in(0,1)$, and since $(1+ \lambda)^{1-\theta}\,Y'\in L^1(0,+\infty)$, then
\[
  |\mathrm{(II)}|\;\leqslant\;\frac{C}{\,|x|^2|y|\,}\Big( \frac{|y|}{\;|x|^2}+\frac{\;|y|^{1-\theta}}{|x|}\Big)\;=\;C\Big( \frac{1}{\;|x|^4}+\frac{1}{\;|x|^3|y|^\theta}\Big)\;\leqslant\frac{2C}{\;|x|^3\,|y|^\theta}
\]
and hence also
\[
 |T(x,y)|\;\leqslant\;\frac{C}{\;|x|^3\,|y|^\theta}\,,\qquad |x|\geqslant 10\,,
\]
for some constant $C>0$ and any $\theta\in(0,1)$. For fixed $p\in(1,\frac{3}{2})$, we take $\theta\in(0,1)$ such that $p'\theta<3$ and $f\in L^p(\Lambda_{1})$: with this choice, H\"{o}lder's inequality yields
\[
\begin{split}
 \|Tf\|_{L^p(\mathbb{R}^3\setminus\Lambda_{10})} 
\;&\leqslant\; C_\theta \,\||x|^{-3}\|_{L^p(\mathbb{R}^3\setminus\Lambda_{10})}\cdot 
\|f\|_{L^p(\Lambda_1)}\cdot \||y|^{-\th}\|_{L^{p'}(\Lambda_1)} \\
&=\;\kappa_p^+\,\|f\|_{L^p(\Lambda_1)}
\end{split}
\]
for some constant $\kappa^+_p>0$. Combining the above bounds yields the boundedness of $T$ as a map from $L^p(\Lambda_1)$ to $L^p(\mathbb{R}^3)$.
\end{proof}

Let us now complete the proof of the $L^p$-boundedness of $W^{+}_{\a, Y}$ for $N=1$ and $p\in(1,\frac{3}{2})$. We only need to show \eqref{eq:needed_one}. Upon re-writing the second equation in \eqref{eq:defO1O2} by means of \eqref{eq:rMu}, that is,
\begin{equation}\label{eq:omega2_rerewritten}
(\Omega_2 u)(x) \;=\; \frac{-1}{(2\pi)^2 |x|^2}
\int_0^{+\infty} e^{-\ii\lambda|x|}F'(\lambda) 
\Big(\int_{\mathbb{R}^3} \frac{e^{\ii\lambda|y|}-e^{-\ii\lambda|y|}}{4\pi |y|} \,
u(y)\,\ud y\Big) \ud\lambda\,,
\end{equation}
it is immediate to recognise that
\begin{equation}\label{eq:O-T-recognised}
 \Omega_2 u\;=\;-(2\pi)^{-2}\,T\!u\,,
\end{equation}
where $T$ is the integral operator given by \eqref{Yint} with $Y\equiv F'$, and $F'$ does satisfy the assumptions of Lemma \ref{lem:kernel-1}. From this, we conclude \eqref{eq:needed_one} at once.

\section{$L^p$-bounds for the general multi-centre case.}\label{sec:Lp-bounds_N}

The additional complication in the case $N\geqslant 2$ is due to the presence, 
in the function $F$ defined in \eqref{eq:def_Fjk} and \eqref{eq:def_F_matrix}, of the 
terms 
$\mathcal{G}_{\lambda}^{y_j y_k}$ (definitions \eqref{eq:def_of_the_Gs}-\eqref{ga-def}), 
which are \emph{oscillatory} in $\lambda$.

Let us start the discussion by re-writing
\begin{equation} \label{eq:G_inverse}
F(\lambda)\;=\; \lambda \,\Ga_{\alpha,Y}(-\lambda)^{-1}\;=\; \lambda 
\Big(\mathcal{A} + \frac{\ii\lambda}{\,4\pi\,}{\mathbbm{1}}- 
\widetilde{G}(-\lambda)\Big)^{-1} , \quad \lambda >0\,,
\end{equation}
with
\begin{align}
  \mathcal{A}\;&:=\;\mathrm{diag}(\alpha_1,\dots,\alpha_N)\,, \\
  \widetilde{G}(\lambda)\;&:=\; 
( \mathcal{G}_{\lambda}^{y_j y_k})_{j,k=1,\dots,N}\,.
\end{align}
We decompose $F(\lambda)$ into a small-$\lambda$ and a large-$\lambda$ part by means 
of two cut-off functions 
$\omega_<$ and $\omega_>$ such that 
\begin{equation}
\begin{split}
\omega_<\;&\in\; C^\infty_0(\mathbb{R})\,,\qquad \omega_{>}(\lambda)\;:=\;1-\omega_<(\lambda)\,, \\
\omega_<(\lambda)\;&=\;
\begin{cases}
\;1 & \textrm{if }\;|\lambda|\,\leqslant\,\gamma \\
\;0 & \textrm{if }\;|\lambda|\,\geqslant\,2\gamma\,,
\end{cases}
\end{split}
\end{equation}
where $\gamma>0$ is a sufficiently large number so that,  
\begin{equation}\label{eq:decompose-F-a}
 \|\mathcal{A}-\widetilde{G}(-\lambda)\|\;<\; |\lambda| (16\pi)^{-1}\,,
\qquad |\lambda|\:\geqslant\:\gamma
\end{equation}
($\|E\|$ being the operator norm of the 
matrix $E$ as an operator on ${\mathbb C}^N$),
and the r.h.s.~of \eqref{eq:G_inverse} is invertible. 
Explicitly,
\begin{equation}\label{eq:decompose-F}
 F\;=\;F^{<}+ F^{>}\,,\qquad F^{<}\;:=\;\w_{<}F\,,\qquad F^{>}\;:=\;\w_{>}F\,.
\end{equation}

From \eqref{eq:G_inverse} and \eqref{eq:decompose-F-a} we expand
\begin{equation}\label{eq:large_lambda_lG}
\begin{split}
 \!\!\!\!F^{>}(\lambda)\;&=\;-4\pi\ii\,\omega_>(\lambda)\,\Big\{\mathbbm{1}-\frac{4\pi\ii}{\lambda}\big(\mathcal{A}-\widetilde{G}(-\lambda)\big)+\Big(\frac{4\pi\ii}{\lambda}\big(\mathcal{A}-\widetilde{G}(-\lambda)\big)\Big)^2\Big\} \\
 &\quad -4\pi\ii\,\omega_>(\lambda)\,\Big(\frac{4\pi\ii}{\lambda}\big(\mathcal{A}-\widetilde{G}(-\lambda)\big)\Big)^3\Big(\mathbbm{1}-\frac{4\pi\ii}{\lambda}\big(\mathcal{A}-\widetilde{G}(-\lambda)\big)\Big)^{\!-1}\,.
\end{split}
\end{equation}
We collect all terms that do not contain 
$\widetilde{G}(-\lambda)$ or for which the oscillation of 
$\widetilde{G}(-\lambda)$ is harmless into the quantity
\begin{equation}
 \begin{split}
 F^{(0)}(\lambda)\;&:=\;F^{<}(\lambda)-4\pi\ii\,\omega_>(\lambda)\,
\Big\{\mathbbm{1}-\frac{4\pi\ii}{\lambda}\mathcal{A}-\frac{16\pi^2}{\lambda^2}\mathcal{A}^2 
 \Big\}  \\
 &\; -4\pi\ii\,\omega_>(\lambda) \Big(\frac{4\pi\ii}{\lambda}\big(\mathcal{A}-\widetilde{G}(-\lambda)\big)\Big)^3\Big(\mathbbm{1}-\frac{4\pi\ii}{\lambda}\big(\mathcal{A}-\widetilde{G}(-\lambda)\big)\Big)^{\!-1} ,
 \end{split}
\end{equation}
whereas
\begin{equation}\label{F-o-1}
 \begin{split}
  F^{(1)}(\lambda)\;&:=\;4\pi\ii\,\omega_>(\lambda) \\
  &\!\!\!\!\!\!\!\!\!\!\!\!\!\!\!\times
\Big\{-\frac{4\pi\ii}{\lambda}\widetilde{G}(-\lambda)-\frac{16\pi^2}{\lambda^2}
\big(\mathcal{A}\,\widetilde{G}(-\lambda)+\widetilde{G}(-\lambda)\,\mathcal{A}\big)
+\frac{16\pi^2}{\lambda^2}\widetilde{G}(-\lambda)^2\Big\}
 \end{split}
\end{equation}
contains the oscillations explicitly.

Thus,
\begin{equation}
 F\;=\;F^{(0)}+F^{(1)}\qquad\textrm{and}\qquad \Omega\;=\;\Omega^{(0)}+\Omega^{(1)}\,,
\end{equation}
where, by means of \eqref{Omega_intermediate},
\begin{equation}\label{Omega_01}
\begin{split}
 (\Omega^{(\ell)}{u})(x)\;&:=\;\frac{1}{\,\ii(2\pi)^{\frac{3}{2}}|x|\,}
\int_0^{+\infty}\!\! e^{-\ii\lambda|x|}
\,F^{(\ell)}(\lambda) \widehat{(rM_u)}(-\lambda)\, \ud\lambda\,, \\
&\quad\ell\in\{0,1\}\,.
\end{split}
\end{equation}

\subsection{$L^p$-boundedness of $\Omega^{(0)}$}

The $L^p$-boundedness of the map $u\mapsto \Omega^{(0)} u$ can be 
established via a straightforward adaptation of the arguments of 
Section \ref{sec:Lp-bounds}, of course understanding that this is 
done for each component $\Omega^{(0)}_{jk}$, and this is possible 
precisely thanks to the lack of relevant oscillations in $\Omega^{(0)}$.

This means that first we write, in analogy to 
\eqref{Omega_intermediate-a},
\begin{equation}\label{Omega_intermediate-a_0}
 (\Omega^{(0)}{u})(x)\;=\;\frac{1}{\,\ii(2\pi)^{\frac{3}{2}}|x|\,}
(\widehat{F^{(0)}} \ast rM_u)(|x|)\,,
\end{equation}
and it is easy to check that $F^{(0)}$ satisfies the properties of the function 
$W$ in Lemma \ref{lem:kernel-0}, from which, reasoning as in \eqref{eq:Opb_final}, 
\begin{equation}\label{eq:O1_b1}
\|\Omega^{(0)} u\|_{L^p(\mathbb{R}^3)}\;\leqslant\;C_p\,\|u\|_{L^p(\mathbb{R}^3)}\,,
\qquad p\in(\textstyle{\frac{3}{2}},3)\,,
\end{equation}
for some constant $C_p>0$.

Then, in analogy to \eqref{eq:O=O1+O2}, \eqref{eq:defO1O2}, 
\eqref{eq:O1rewritten}, \eqref{eq:O2_rewritten}, and 
\eqref{eq:omega2_rerewritten}, we split
\begin{equation}
\Omega^{(0)} u \;=\;\Omega^{(0)}_1 u+\Omega^{(0)}_2 u
\end{equation}
with
\begin{equation}\label{eq:defO1}
 \begin{split}
  (\Omega^{(0)}_1 u)(x)\;&:=\;\frac{-\ii}{(2\pi)^{\frac{3}{2}}|x|^2 }
\int_0^{+\infty}\!\! e^{-\ii\lambda |x| }
\,F^{(0)}(\lambda) \widehat{(r^2 M_u)}(-\lambda)\, \ud\lambda\,, \\
&=\; \frac{-\ii}{(2\pi)^{\frac{3}{2}}|x|^2 } 
(\widehat{F^{(0)}}\ast r^2 M_u)(|x|)
 \end{split}
\end{equation}
and
\begin{equation}\label{eq:defO2}
 \begin{split}
 (&\Omega^{(0)}_2 u)(x)\;:=\;\frac{-1}{(2\pi)^{\frac{3}{2}}|x|^2 }
\int_0^{+\infty}\!\! e^{-\ii\lambda|x| }
\,{F^{(0)}}'(\lambda) \widehat{(r M_u)}(-\lambda)\ \ud\lambda \\
&=\;\frac{-1}{(2\pi)^{\frac{3}{2}}\r^2 }
(\mathcal{L}^{(0)} * r M_u)(\r) \\
&=\;\frac{-1}{(2\pi)^2 |x|^2}
\int_0^{+\infty} \!\!e^{-\ii\lambda|x|}{F^{(0)}}'(\lambda) 
\Big(\int_{\mathbb{R}^3} \frac{e^{\ii\lambda|y|}-e^{-\ii\lambda|y|}}{4\pi |y|} \,
u(y)\,\ud y\Big) \ud\lambda\,,
 \end{split}
\end{equation}
where $\mathcal{L}^{(0)}$ is the Fourier transform of the function 
${\bf 1}_{(0,\infty)}{F^{(0)}}'$.

Since, as observed already, $F^{(0)}$ behaves like $W$ in Lemma \ref{lem:kernel-0}, 
we have, reasoning as in \eqref{eq:bound_on_O1},
\begin{equation}\label{eq:O1_b2}
\|\Omega^{(0)}_1 u\|_{L^p(\mathbb{R}^3)}\;\leqslant\;C_p\,\|u\|_{L^p(\mathbb{R}^3)}\,,
\qquad p\in(\textstyle{1,\frac{3}{2}})\,,
\end{equation}
and since ${\bf 1}_{(0,\infty)}{F^{(0)}}'$ too satisfies the properties of the 
function $W$ in Lemma \ref{lem:kernel-0}, we have, using the \emph{second} line 
in the r.h.s.~of \eqref{eq:defO2} and reasoning as in 
\eqref{eq:O2partial}-\eqref{eq:first_bound_on_O2},
\begin{equation}\label{eq:O1_b3}
\|\Omega^{(0)}_2 \mathbf{1}_{\{|x|\geqslant 1\}} u\|_{L^p(\mathbb{R}^3)}\;
\leqslant\;C_p\,\|\mathbf{1}_{\{|x|\geqslant 1\}} u\|_{L^p(\mathbb{R}^3)}\,,
\qquad p\in(\textstyle{1,\frac{3}{2}})\,,
\end{equation}
for some constant $C_p>0$. Last, since ${\bf 1}_{(0,\infty)}{F^{(0)}}'$ satisfies 
the properties of the function $Y$ in Lemma \ref{lem:kernel-1}, we have, using the 
\emph{third} line in the r.h.s.~of \eqref{eq:defO2} and reasoning as in 
\eqref{eq:omega2_rerewritten}-\eqref{eq:O-T-recognised} and \eqref{eq:needed_one},
\begin{equation}\label{eq:O1_b4}
\|\Omega^{(0)}_2 \mathbf{1}_{\{|x|\leqslant 1\}} u\|_{L^p(\mathbb{R}^3)}\;
\leqslant\;C_p\,\|\mathbf{1}_{\{|x|\leqslant 1\}} u\|_{L^p(\mathbb{R}^3)}\,,
\qquad p\in(\textstyle{1,\frac{3}{2}})\,.
\end{equation}

Combining together the bounds \eqref{eq:O1_b1}, \eqref{eq:O1_b2}, 
\eqref{eq:O1_b3}, and \eqref{eq:O1_b4}, plus interpolation 
so as to cover also the case $p=\frac{3}{2}$, yields finally
\begin{equation}
 \|\Omega^{(0)}u\|_{L^p(\mathbb{R}^3)}\;\leqslant\;C_p\,\|u\|_{L^p(\mathbb{R}^3)}\,,\qquad p\in(1,3)\,,
\end{equation}
for some constant $C_p>0$.

\subsection{$L^p$-boundedness of $\Omega^{(1)}$}

The proof of the $L^p$-boundedness of the map $u\mapsto \Omega^{(1)} u$ is 
somewhat more involved, however the basic idea of the proof is similar to that for 
$\Omega^{(0)}$. 
First  we re-write \eqref{Omega_01} in analogy to \eqref{Omega_intermediate} 
and \eqref{Omega_intermediate-a_0} as
\begin{equation}\label{Omega_intermediate-a_1}
 (\Omega^{(1)}{u})(x)\;=\;\frac{1}{\,\ii(2\pi)^{\frac{3}{2}}|x|\,}
(\widehat{F^{(1)}} \ast rM_u)(|x|)\,.
\end{equation}
Owing to \eqref{F-o-1}, the matrix elements of $F^{(1)}(\lambda)$ entering 
\eqref{Omega_01} and \eqref{Omega_intermediate-a_1} above are of the form
\begin{equation*}
 \frac{\omega_>(\lambda)}{\lambda}\,
\frac{e^{-\ii \lambda|y_j-y_k|}}{\,|y_j-y_k|\,}\,,\qquad 
\frac{\omega_>(\lambda)}{\lambda^2}\,\frac{e^{-\ii \lambda|y_j-y_k|}}{\,|y_j-y_k|\,}\,,
\qquad \frac{\omega_>(\lambda)}{\lambda^2}\,
\frac{e^{-\ii \lambda|y_j-y_k|}}{\,|y_j-y_k|\,}
\frac{e^{-\ii \lambda|y_r-y_s|}}{\,|y_r-y_s|\,}
\end{equation*}
(observe that the $\lambda$-dependence of the matrix elements of 
$\widetilde{G}$ in \eqref{F-o-1} is $\widetilde{G}(-\lambda)$). This means that 
denoting by $a>0$ any of the numbers $|y_j-y_k|$ or $|y_j-y_k|+|y_r-y_s|$ and by 
$X(\lambda)$ the function $\lambda^{-1}\omega_>(\lambda)$ or 
$\lambda^{-2}\omega_>(\lambda)$, formulas \eqref{Omega_01} and 
\eqref{Omega_intermediate-a_1} imply that $\Omega^{(1)}{u}$ is a linear 
combination of terms of the form
\begin{equation}\label{eq:def_Xi}
\begin{split}
 (\Xi\,u)(x)\,:=&\;\frac{1}{\,\ii\,|x|\,}
\int_0^{+\infty}\!\! e^{-\ii\lambda(|x|+a)}
X(\lambda)\, \widehat{(rM_u)}(-\lambda)\, \ud\lambda \\
=&\;\frac{1}{\,\ii\,|x|\,}
(\widehat{X} \ast rM_u)(|x|+a)\,,
\end{split}
\end{equation}
and we need to prove the $L^p$-boundedness of the map $u\mapsto\Xi\,u$. 
In fact, we shall establish it for each of the two terms of the bound
\begin{equation}\label{eq:Xi_splitted}
 \|\Xi\,u\|_{L^p(\mathbb{R}^3)}\;\leqslant\;\|\mathbf{1}_{\{|x|\geqslant R\}}\,\Xi\,u\|_{L^p(\mathbb{R}^3)}+\|\mathbf{1}_{\{|x|\leqslant R\}}\,\Xi\,u\|_{L^p(\mathbb{R}^3)}
\end{equation}
for a suitable $R>0$. 

Let us cast the discussion of such two terms into the following two Lemmas. The combination of \eqref{eq:Xi_splitted} above with \eqref{eq:Xi_ext} and \eqref{eq:Xi_int} below will then complete the proof of the $L^p$-boundedness of $\Omega^{(1)}$.

\begin{lem}
For any $p\in(1,3)$ and $R>a$ there exists a constant $C_p>0$ such that
\begin{equation}\label{eq:Xi_ext}
 \|\mathbf{1}_{\{|x|\geqslant R\}}\,\Xi\,u\|_{L^p(\mathbb{R}^3)}\;
\leqslant\;C_p\,\|u\|_{L^p(\mathbb{R}^3)}
\end{equation}
for all $u\in L^p(\mathbb{R}^3)$, where $\Xi\,u$ is defined in \eqref{eq:def_Xi}.
\end{lem}

\begin{proof}
We consider first the case $p\in(\frac{3}{2},3)$. From \eqref{eq:def_Xi} and from 
the fact that $\rho\geqslant  R+a$ implies 
$\frac{1}{2}\rho\leqslant\rho-a\leqslant\rho$,
\[
 \begin{split}
  \|\mathbf{1}_{\{|x|\geqslant R\}}\,\Xi\,u\|^p_{L^p(\mathbb{R}^3)}\; &=\;
4\pi\!\int_R^{+\infty}\rho^{2-p}\big|(\widehat{X} \ast rM_u)(\rho+a)\big|^p\,\ud \rho \\
  &=\;4\pi\!\int_{ R+a}^{+\infty}(\rho-a)^{2-p}
\big|(\widehat{X} \ast rM_u)(\rho)\big|^p\,\ud \rho \\
  &\leqslant\;C_p\int_{0}^{+\infty}
\big|(\widehat{X} \ast rM_u)(\rho)\big|^p\,\rho^{2-p}\ud \rho\,.
 \end{split}
\]
Now, $\rho^{2-p}$ is an $A_p$-weight on $\mathbb{R}$ because $p\in(\frac{3}{2},3)$ 
(Theorem \ref{thm:harmonic_analysis_results}(iv)) and the convolution with 
$\widehat{X}$ is a Calder\'on-Zygmund operator on $\mathbb{R}$ because the function 
$X$ obviously satisfies the properties of the function $W$ in Lemma \ref{lem:kernel-0}. 
Then it follows from Theorem \ref{thm:harmonic_analysis_results}(ii) that
\begin{equation*}
\|\mathbf{1}_{\{|x|\geqslant R\}}\,\Xi\,u\|_{L^p(\mathbb{R}^3)}^p \;\leqslant\; 
C_p \int_0^{+\infty} |(\rho M_u)(\r)|^p \,\r^{2-p}\,\ud\r\;\
\leqslant\;C_p'\|u\|_{L^p(\mathbb{R}^3)}^p
\end{equation*}
for suitable $C_p'>0$. The Lemma is then proved in the case $p\in(\frac{3}{2},3)$.

Next we consider the case $p\in(1,\frac{3}{2})$. Integration by parts in 
\eqref{eq:def_Xi}, using 
$e^{-\ii\lambda(\r+a)}=\ii(\r+a)^{-1}\partial_\lambda 
e^{-\ii\lambda(\r+a)}$, yields   
\begin{equation*}
 \Xi\,u\;=\;\Xi_1u + \Xi_2u
\end{equation*}
with
\begin{equation*}
\begin{split}
(\Xi_1u)(x)\,:=&\;\frac{-\ii}{\,|x|(|x|+a)\,}
\int_0^{+\infty}\!\! e^{-\ii\lambda(|x|+a)}
X(\lambda)\, \widehat{(r^2M_u)}(-\lambda)\, \ud\lambda \\
=&\;\frac{-\ii}{\,|x|(|x|+a)\,}
(\widehat{X} \ast r^2M_u)(|x|+a)
\end{split}
\end{equation*}
and
\begin{equation*}
\begin{split}
(\Xi_2u)(x)\,:=&\;\frac{-1}{\,|x|(|x|+a)\,}
\int_0^{+\infty}\!\! e^{-\ii\lambda(|x|+a)}
X'(\lambda)\, \widehat{(rM_u)}(-\lambda)\, \ud\lambda \\
&=\;\frac{-1}{\,|x|(|x|+a)\,}
(\widehat{X'} \ast rM_u)(|x|+a)\,.
\end{split}
\end{equation*}
Up to a change of variable, the quantity 
$\|\mathbf{1}_{\{|x|\geqslant R\}}\,\Xi_1u\|_{L^p(\mathbb{R}^3)}$ is 
estimated precisely as the quantity $\|\Omega_1 u\|_{L^p(\mathbb{R}^3)}$ in 
Section \ref{subsec:N1_p132} -- see \eqref{eq:bound_on_O1} above. Indeed,
\[
 \begin{split}
  \|\mathbf{1}_{\{|x|\geqslant R\}}\,\Xi_1u\|_{L^p(\mathbb{R}^3)}^p\;&=\,\int_R^{+\infty}\!\!\frac{4\pi\rho^2}{\,\rho^p(\rho+a)^p}\,|(\widehat{X}*r^2M_u)(\rho+a)|^p\,\ud \rho \\
  &=\,\int_{R+a}^{+\infty}\!\!4\pi(\rho-a)^{2-p}\rho^{-p}\,
|(\widehat{X}*r^2M_u)(\rho)|^p\,\ud \rho \\
  &\leqslant\;C\!\int_{0}^{+\infty}\!\!|(\widehat{X}*r^2M_u)(\rho)|^p\,
\rho^{2-2p}\,\ud\rho \\
  &\leqslant\;C\!\int_{0}^{+\infty}\!\!|M_u(\rho)|^p\,\rho^{2}\,\ud\rho\;
\leqslant\;C_p\|u\|_{L^p(\mathbb{R}^3)}^p
 \end{split}
\]
for some constants $C,C_p>0$,
having used $\frac{1}{2}\rho \leqslant \rho-a\leqslant\rho$ in the third step 
and Theorem \ref{thm:harmonic_analysis_results}(ii) in the fourth step. This was 
possible because $\rho^{2-2p}$ is an $A_p$-weight on $\mathbb{R}$ for 
$p\in(1,\frac{3}{2})$ (Theorem \ref{thm:harmonic_analysis_results}(iv)) and 
because $f\mapsto \widehat{X}*f$ is a Calder\'on-Zygmund operator on 
$\mathbb{R}$ (the function $X$ does satisfy the assumptions on the function 
$W$ in Lemma \ref{lem:kernel-0}).

It remains to estimate the quantity 
$\|\mathbf{1}_{\{|x|\geqslant R\}}\,\Xi_2u\|_{L^p(\mathbb{R}^3)}$ 
in the regime $p\in(1,\frac{3}{2})$ and we proceed by splitting
\[
\begin{split}
 &\|\mathbf{1}_{\{|x|\geqslant R\}}\,\Xi_2u\|_{L^p(\mathbb{R}^3)}^p\;= \\
 &\qquad=\;\|\mathbf{1}_{\{|x|\geqslant R\}}\,\Xi_2\mathbf{1}_{\{|x|\geqslant R\}}u\|_{L^p(\mathbb{R}^3)}^p+\|\mathbf{1}_{\{|x|\geqslant R\}}\,\Xi_2\mathbf{1}_{\{|x|\leqslant R\}}u\|_{L^p(\mathbb{R}^3)}^p\,.
\end{split}
 \]
For estimating 
$\|\mathbf{1}_{\{|x|\geqslant R\}}\,\Xi_2 
\mathbf{1}_{\{|x|\geqslant R\}}u\|_{L^p(\mathbb{R}^3)}$ 
we observe that
\begin{align}
  \|\mathbf{1}_{\{|x|\geqslant R\}}\,\Xi_2u\|_{L^p(\mathbb{R}^3)}^p\;
&=\;\int_R^{+\infty}\!\!\frac{4\pi\rho^2}{\,\rho^p(\rho+a)^p}\,
|(\widehat{X'}*rM_u)(\rho+a)|^p\,\ud \rho  \notag \\
  &=\;\int_{R+a}^{+\infty}4\pi(\rho-a)^{2-p}\rho^{-p}\,|(\widehat{X'}*rM_u)(\rho)|^p\,\ud \rho \notag \\
  &  
\leqslant\;C\int_{0}^{+\infty}\!\!|(\widehat{X'}*rM_u)(\rho)|^p\,\rho^{2-2p}\,\ud\rho \tag{*}
\end{align}
for some constant $C>0$, where we used again 
$\frac{1}{2}\rho< \rho-a \leqslant\rho$. Then 
we can proceed exactly as in \eqref{eq:O2partial}-\eqref{eq:first_bound_on_O2}, 
because  $\rho^{2-2p}$ is an $A_p$-weight on $\mathbb{R}$ for $p\in(1,\frac{3}{2})$ and 
 $f\mapsto \widehat{X'}*f$ is a Calder\'on-Zygmund operator on 
$\mathbb{R}$; the conclusion is the same as in \eqref{eq:first_bound_on_O2}, that is,
\[
 \|\mathbf{1}_{\{|x|\geqslant R\}}\,\Xi_2\mathbf{1}_{\{|x|\geqslant R\}}u\|_{L^p(\mathbb{R}^3)}^p\;\leqslant\;C_p\,\|\mathbf{1}_{\{|x|\geqslant R\}}u\|_{L^p(\mathbb{R}^3)}^p
\]
for some constant $C_p>0$. We also observe from (*)
that
\begin{multline*}
  \|\mathbf{1}_{\{|x|\geqslant R\}}\,\Xi_2u\|_{L^p(\mathbb{R}^3)}^p\; \\
\leqslant\;\int_{\mathbb{R}^3}\ud x\:\Big|\,\frac{1}{\;|x|^2}
\int_0^{+\infty}\!\!e^{-\ii\lambda|x|}\,X'(\lambda)\,
\widehat{(r M_u)}(-\lambda)\,\ud\lambda\,\Big|^p\;
=\;\|\widetilde{\Xi}_2u\|_{L^p(\mathbb{R}^3)}^p\,,
\end{multline*} 
where $\widetilde{\Xi}_2u$ has precisely the same structure as 
$\Omega_2 u$ in \eqref{eq:defO1O2} with the function $X'$ here in place of the 
function $F'$ therein. Therefore, as argued in 
\eqref{eq:omega2_rerewritten}-\eqref{eq:O-T-recognised}, 
since $X'$ satisfies the assumptions on the function $Y$ in Lemma \ref{lem:kernel-1}, 
the conclusion is the same as in \eqref{eq:needed_one}, that is,
\[
 \|\mathbf{1}_{\{|x|\geqslant R\}}\,
\Xi_2\mathbf{1}_{\{|x|\leqslant R\}}u\|_{L^p(\mathbb{R}^3)}^p\;
\leqslant\;C_p\,\|\mathbf{1}_{\{|x|\leqslant R\}}u\|_{L^p(\mathbb{R}^3)}^p
\]
for some constant $C_p>0$. Therefore,
\[
 \|\mathbf{1}_{\{|x|\geqslant R\}}\,\Xi_2u\|_{L^p(\mathbb{R}^3)}^p\;
\leqslant\;C_p\,\|u\|_{L^p(\mathbb{R}^3)}^p
\]
and Lemma is then proved in the case $p\in(1,\frac{3}{2})$.

Last, by interpolation the Lemma is also proved in the case $p=\frac{3}{2}$.
\end{proof}

\begin{lem}
For any $p\in(1,3)$ and $R>100a$ there exists a constant $C_p>0$ such that
\begin{equation}\label{eq:Xi_int}
 \|\mathbf{1}_{\{|x|\leqslant R\}}\,\Xi\,u\|_{L^p(\mathbb{R}^3)}\;\leqslant\;C_p\,\|u\|_{L^p(\mathbb{R}^3)}
\end{equation}
for all $u\in L^p(\mathbb{R}^3)$, where $\Xi\,u$ is defined in \eqref{eq:def_Xi}.
\end{lem}

\begin{proof}
By means of \eqref{eq:rMu} we see that the map 
$u\mapsto\Xi\,u$ defined in \eqref{eq:def_Xi} is an integral operator with 
kernel $\frac{\ii}{4\pi}\,K_\Xi(x,y)$ given by
\begin{equation}\label{eq:kernelK}
 K_\Xi(x,y)\;:=\frac{1}{\,\sqrt{2\pi}\,}
 \int_0^{+\infty} \frac{e^{-\ii\lambda(|x|+a)}(e^{-\ii\lambda|y|}-e^{\ii\lambda |y|})}
{|x|\,|y|}\,X(\lambda)\,\ud \lambda\,.
\end{equation}
Since $X(\lambda)=\lambda^{-1}\omega_>(\lambda)$ or $\lambda^{-2}\omega_>(\lambda)$, 
obviously $\rho\mapsto\widehat {X}(\r)$ is smooth for $\r\not=0$  and with rapid 
decrease as $\r \to +\infty$. Moreover, since $X \in L^q(\mathbb{R})$ for any 
$q>1$,  $\widehat{X} \in L^{p}(\mathbb{R})$ for any $p\in[2,\infty)$, 
owing to the Hausdorff-Young inequality.  
Thus, $\widehat{X} \in L^{p}(\mathbb{R})$ for any $p\in[1,\infty)$.

We shall prove the Lemma by splitting
\begin{equation}
\begin{split}
 &\!\!\!\|\mathbf{1}_{\{|x|\leqslant R\}}\,\Xi\,u\|_{L^p(\mathbb{R}^3)}^p\;  \\
 &=\;\|\mathbf{1}_{\{|x|\leqslant R\}}\,\Xi\,
\mathbf{1}_{\{|x|\geqslant 10 R\}}u\|_{L^p(\mathbb{R}^3)}^p+
\|\mathbf{1}_{\{|x|\leqslant R\}}\,\Xi\,
\mathbf{1}_{\{|x|\leqslant 10 R\}}u\|_{L^p(\mathbb{R}^3)}^p
\end{split}
 \end{equation}
and estimating separately the two summands in the r.h.s.~above.

When $R>100 a$, $|x|\leqslant R$, and $|y|\geqslant 10R$, one has 
$|\widehat{X}(|x|\pm|y|+a)|\leqslant C_n\langle y\rangle^{-n}$ for any 
$n\in\mathbb{N}$ and suitable constants $C_n>0$, which follows from 
the rapid decrease of $\widehat{X}$. Then the identity
\begin{equation}\label{eq:Xhat}
 K_\Xi(x,y)\;=\;\frac{\,\widehat{X}(|x|+a+|y|)-\widehat{X}(|x|+a-|y|)\,}{|x||y|}
\end{equation}
shows that in this regime 
$|K_\Xi(x,y)|\leqslant 2\,C_n |x|^{-1}|y|^{-1}\langle y\rangle^{-n}$. 
Therefore, for any  $p\in (1,3)$ and corresponding $n$ large enough,
\[
 \|K_\Xi(x,\cdot)\mathbf{1}_{\{|\cdot|\geqslant 10R\}}\|_{L^{p'}(\mathbb{R}^3)}\;
\leqslant\;\frac{\,2C_n\,}{|x|}\Big(\int_{|y|\geqslant 10R}\frac{1}{\,|y|^{p'}
\langle y\rangle^{np'}}\Big)^{\!1/{p'}}\;
\leqslant\;\frac{C_p}{|x|}
\]
for some constant $C_p>0$. The latter bound and H\"{o}lder's inequality then yield, 
for any $p\in(1,3)$, 
\begin{equation}\label{eq:p1}
 \begin{split}
  \|&\mathbf{1}_{\{|\cdot|\leqslant R\}}\,\Xi\,
\mathbf{1}_{\{|\cdot|\geqslant 10 R\}}u\|_{L^p(\mathbb{R}^3)}^p\; \\
  &\quad\leqslant\;\int_{\mathbb{R}^3}\ud x\,\mathbf{1}_{\{|x|\leqslant R\}}\!(x)\,
\Big|\int_{\mathbb{R}^3}\ud y\,K_\Xi(x,y)\mathbf{1}_{\{|y|\geqslant 10 R\}}\!(y)\,u(y)
\Big|^{p} \\
  &\quad\leqslant\;C_p\,
\Big\|\frac{\,\mathbf{1}_{\{|x|\leqslant R\}}}{|x|}\Big\|_{L^p(\mathbb{R}^3)}^p\,
\|\mathbf{1}_{\{|\cdot|\geqslant 10 R\}}u\|_{L^p(\mathbb{R}^3)}^p \\
  &\quad=\;C_p'\,\|\mathbf{1}_{\{|\cdot|\geqslant 10 R\}}u\|_{L^p(\mathbb{R}^3)}^p
 \end{split}
\end{equation}
for some constant $C_p'>0$.

This provides the first partial estimate for the proof of \eqref{eq:Xi_int}: the proof is completed when we show in addition that
\begin{equation}\label{eq:p2}
 \|\mathbf{1}_{\{|\cdot|\leqslant R\}}\,\Xi\,\mathbf{1}_{\{|\cdot|\leqslant 10 R\}}u\|_{L^p(\mathbb{R}^3)}^p\;\leqslant\;C_p\,\|\mathbf{1}_{\{|\cdot|\leqslant 10 R\}}u\|_{L^p(\mathbb{R}^3)}^p
\end{equation}
for any $p\in(1,3)$ and suitable constant $C_p>0$. We shall establish \eqref{eq:p2} above in three separate regimes: $p\in(2,3)$, $p\in(\frac{3}{2},2)$, and $p\in(1,\frac{3}{2})$. By interpolation, also the cases $p=\frac{3}{2}$ and $p=2$ will then be covered.

From \eqref{eq:Xhat} we estimate
\begin{equation}\label{eq:Kpprimenorm}
 \begin{split}
  &\|K_\Xi(x,\cdot)\mathbf{1}_{\{|\cdot|\leqslant 10R\}}\|_{L^{p'}(\mathbb{R}^3)}\;
\\
  &\qquad\leqslant\;\frac{\:(4\pi)^{\frac{1}{p'}}}{|x|}\sum_{\pm}
\Big(\int_0^{10R}\!\ud\rho\,\rho^{2-p'}\,
|\widehat{X}(|x|+a\pm\rho)|^{p'}\Big)^{\!1/{p'}}. \end{split}
\end{equation}
When $p\in(2,3)$, and hence $p'\in(\frac{3}{2},2)$, we have 
$\rho^{2-p'}\leqslant(10R)^{2-p'}$ for every $\rho\in[0,10R]$, 
and \eqref{eq:Kpprimenorm} then yields
\begin{equation}\label{eq:estK2-3}
 \|K_\Xi(x,\cdot)\mathbf{1}_{\{|\cdot|\leqslant 10R\}}\|_{L^{p'}(\mathbb{R}^3)}\;\leqslant \;C\,\frac{\;\|\widehat{X}\|_{L^{p'}(\mathbb{R})}}{|x|}\,.
\end{equation}
When instead $p\in(\frac{3}{2},2)$, and hence $p'\in(2,3)$, the 
r.h.s.~of \eqref{eq:Kpprimenorm} is estimated with H\"{o}lder's inequality, 
with weights $q=\frac{p'-1}{2(p'-2)}$ and $q'=\frac{p'-1}{3-p'}$, as
\begin{equation}\label{eq:estK32-2}
\begin{split}
 \|K_\Xi(x,\cdot)&\mathbf{1}_{\{|\cdot|\leqslant 10R\}}\|_{L^{p'}(\mathbb{R}^3)}\; \\
 &\leqslant\;\frac{C}{|x|}\,
\Big(\int_0^{10R}\frac{\ud \rho}{\,\rho^{\frac{p'-1}{2}}}\Big)^{\! \frac{2(2-p)}{p'}}\,
\|\widehat{X}\|_{L^{\!\frac{p'(p'-1)}{3-p'}}\!(\mathbb{R})}\,.
 \end{split}
\end{equation}

In order to obtain analogous estimates to \eqref{eq:estK2-3}-\eqref{eq:estK32-2} 
in the remaining regime $p\in(1,\frac{3}{2})$, it is convenient to integrate by 
parts in \eqref{eq:kernelK}, using 
$e^{-\ii\lambda(|x|+a)}=\ii(|x|+a)^{-1}\partial_\lambda e^{-\ii\lambda(|x|+a)}$, 
so as to split
\begin{equation}\label{eq:KK1K2}
 K_\Xi(x,y)\;=\;K_\Xi^{(1)}(x,y)+K_\Xi^{(2)}(x,y)
\end{equation}
with
\begin{equation}\label{eq:defKXi1}
 \begin{split}
  K_\Xi^{(1)}&(x,y) \\
  :=&\;\frac{-1}{\sqrt{2\pi}\,|x|(|x|+a)}
\int_0^{+\infty}\!\!(e^{-\ii\lambda(|x|+a+|y|)}+
e^{\ii\lambda(|x|+a-|y|})\,X(\lambda)\,\ud\lambda \\
  =&\;\frac{-1}{\,|x|(|x|+a)}\,
\big(\widehat{X}(|x|+a+|y|)+\widehat{X}(|x|+a-|y|)\big)
 \end{split}
\end{equation}
and
\begin{equation}\label{eq:defKXi2}
 \begin{split}
  &K_\Xi^{(2)}(x,y) \\
  &\quad:=\;\frac{-\ii}{\sqrt{2\pi}\,|x|(|x|+a)|y|}
\int_0^{+\infty}\!\!\!e^{-\ii\lambda(|x|+a)}
(e^{-\ii\lambda|y|}-e^{\ii\lambda|y|})\,X'(\lambda)\,\ud\lambda\,.
 \end{split}
\end{equation}

Using \eqref{eq:defKXi1} we get
\begin{equation}\label{eq:estK32-3}
 \begin{split}
  &\|K_\Xi^{(1)}(x,\cdot)\mathbf{1}_{\{|\cdot|\leqslant 10R\}}\|_{L^{p'}(\mathbb{R}^3)}\\
&\qquad\leqslant\;\frac{\:(4\pi)^{\frac{1}{p'}}}{\,|x|(|x|+a)}\sum_{\pm}
\Big(\int_0^{10R}\!\ud\rho\,\rho^{2}\,|\widehat{X}(|x|+a\pm\rho)|^{p'}\Big)^{\!1/{p'}} \\
&\qquad\leqslant\;\frac{\,2\,(4\pi)^{\frac{1}{p'}}(10R)^2}{\,|x|(|x|+a)}\,
\|\widehat{X}\|_{L^{p'}(\mathbb{R})}\,.
 \end{split}
\end{equation}
As for $K_\Xi^{(2)}$, we exploit \eqref{eq:defKXi2} using the 
bound $|X'(\lambda)|\leqslant C\langle\lambda\rangle^{ -2}$ for some 
$C>0$, which follows 
from the fact that $X(\lambda)=\lambda^{-1}\omega_>(\lambda)$ or 
$\lambda^{-2}\omega_>(\lambda)$,
and the bound $|e^{-\ii\lambda|y|}-e^{\ii\lambda|y|}|\leqslant 2(\lambda|y|)^{1-\theta}$ 
$\forall\theta\in(0,1)$. 
Thus,
\[
 |K_\Xi^{(2)}(x,y)|\;\leqslant\;\frac{1}{\,|x|(|x|+a)|y|}
\int_0^{+\infty}\!\!2(\lambda|y|)^{1-\theta}\,| X'(\lambda)|\,\ud\lambda\;
\leqslant\;\frac{C}{\,|x|(|x|+a)|y|^\theta}\,,
\]
whence
\begin{equation}\label{eq:estK32-4}
 \|K_\Xi^{(2)}(x,\cdot)\mathbf{1}_{\{|\cdot|\leqslant 10R\}}\|_{L^{p'}(\mathbb{R}^3)}\;\leqslant\; \frac{C'}{\,|x|(|x|+a)}\,\Big\|\,\frac{\mathbf{1}_{\{|y|\leqslant 10R\}}}{|y|^\theta}\Big\|_{L^{p'}(\mathbb{R}^3)}\,,
\end{equation}
for suitable constants $C,C'>0$, 
where the $L^{p'}$-norm in the r.h.s.~is finite whenever $ \theta p'<3$.

The estimates \eqref{eq:estK2-3}, \eqref{eq:estK32-2}, \eqref{eq:KK1K2}, \eqref{eq:estK32-3}, and \eqref{eq:estK32-4} together then imply that, for some constant $C_p>0$,
\begin{equation}\label{eq:anotherboundK}
 \|K_\Xi(x,\cdot)\mathbf{1}_{\{|\cdot|\leqslant 10R\}}\|_{L^{p'}(\mathbb{R}^3)}\;\leqslant \;\frac{C_p}{|x|}\,,\qquad p\in(\textstyle{1,\frac{3}{2}})\cup(\textstyle{\frac{3}{2}},2)\cup(2,3)\,.
\end{equation}
Then \eqref{eq:anotherboundK} and H\"{o}lder's inequality yield
\begin{equation}
 \begin{split}
  \|&\mathbf{1}_{\{|\cdot|\leqslant R\}}\,\Xi\,\mathbf{1}_{\{|\cdot|\leqslant 10 R\}}u\|_{L^p(\mathbb{R}^3)}^p \\
  &\quad\leqslant\;\int_{\mathbb{R}^3}\ud x\,\mathbf{1}_{\{|x|\leqslant R\}}\!(x)\,
\Big|\int_{\mathbb{R}^3}\ud y\,K_\Xi(x,y)
\mathbf{1}_{\{|y|\leqslant 10 R\}}\!(y)\,u(y)\Big|^{p} \\
  &\quad\leqslant\;C_p\,\Big\|\frac{\,\mathbf{1}_{\{|x|\leqslant R\}}}{|x|}\Big\|_{L^p(\mathbb{R}^3)}^p\,\|\mathbf{1}_{\{|\cdot|\leqslant 10 R\}}u\|_{L^p(\mathbb{R}^3)}^p \\
  &=\;C_p'\,\|\mathbf{1}_{\{|\cdot|\leqslant 10 R\}}u\|_{L^p(\mathbb{R}^3)}^p
 \end{split}
\end{equation}
for some constant $C_p'>0$.

We have thus obtained precisely the desired estimate \eqref{eq:p2}. 
This completes the proof because, as commented already,  \eqref{eq:p1} 
and \eqref{eq:p2} together give \eqref{eq:Xi_int}.
\end{proof}

\section{Unboundedness in $L^1(\mathbb{R}^3)$ and 
$L^p(\mathbb{R}^3)$, $p\geqslant 3$.}\label{sec:unbound}

In this Section we complete the proof of Theorem \ref{ref:main_thm} as far as 
the unboundedness part is concerned, hence showing that the wave operators 
$W^\pm_{\alpha, Y}$ 
are unbounded in $L^p(\mathbb{R}^3)$ whenever $p\in\{1\}\cup[3,+\infty]$. 
As commented already at the beginning of Section \ref{sec:Lp-bounds}, it is enough 
to prove this property for $W^{+}_{\alpha,Y}$: then the same conclusion 
follows for $W^{-}_{\alpha,Y}$.

\subsection{Unboundedness of $W^{+}_{\alpha,Y}$ in 
$L^p(\mathbb{R}^3)$ for $p\in[3,+\infty]$}

Because of the $L^p$-boundedness of $W^{+}_{\alpha,Y}$ for $p\in(1,3)$, it is clear 
that we only need to prove that $W^{+}_{\alpha,Y}$ is unbounded in $L^3(\mathbb{R}^3)$, 
for any $L^p$-boundedness for $p>3$ would then  contradict, by interpolation, 
the unboundedness when $p=3$.

Let us assume for contradiction that $W_{\alpha,Y}^+$ is bounded 
in $L^3(\mathbb{R}^3)$, which by duality implies also that 
$(W^+_{\alpha,Y})^*$ is bounded in $L^{3/2}(\mathbb{R}^3)$.

Theorem \ref{thm:general_properties}(iv) guarantees that we may choose 
$c>0$ sufficiently large so as to make the matrix $\Ga_{\alpha,Y}(\ii c)$ 
non-singular. Correspondingly, 
$R_0(-c^2)=(H_0+c^2\mathbbm{1})^{-1}$ maps continuously $L^{3/2}(\mathbb{R}^3)$ into 
$W^{2,3/2}(\mathbb{R}^3)$ and hence also $L^{3/2}(\mathbb{R}^3)$ into 
$L^q(\mathbb{R}^3)$ for any $q\in[\frac{3}{2},\infty)$, 
owing to a Sobolev embedding.

Thus, the $L^{3/2}$-boundedness of $(W^+_{\alpha,Y})^*$, the 
$L^{3/2}\to L^3$-boundedness of $R_0(-c^2)$, and the $L^3$-boundedness of 
$W^+_{\alpha,Y}$ imply, by means of the intertwining property \eqref{eq:intertwining}, 
that also the operator
\[
R_{\alpha,Y}(-c^2)P_{\mathrm{ac}}(H_{\alpha,Y})\;
=\;W_{\alpha,Y}^+R_{0}(-c^2)(W_{\alpha,Y}^+)^\ast
\]
is continuous from $L^{3/2}(\mathbb{R}^3)$ to $L^3(\mathbb{R}^3)$.
As a consequence, we read out from the the resolvent identity 
\eqref{eq:resolvent_identity} that for any 
$u\in L^2_{\mathrm{ac}}(H_{\alpha,Y})\cap L^{3/2}(\mathbb{R}^3)$ 
the function
\[\tag{*}
\begin{split}
R_{\alpha,Y}(-c^2) &u - R_0(-c^2)u \\
&=\; 
\sum_{j,k=1}^N (\Gamma_{\alpha,Y}(\ii c)^{-1})_{jk}
\;\mathcal{G}_{\ii c}^{y_j}(x) 
\!\int_{\mathbb{R}^3} \mathcal{G}_{\ii c}^{y_k}(y)\,u(y)\,\ud y
\end{split}
\]
must belong to $L^3(\mathbb{R}^3)$. 

Let us make now a choice of $u$ for which the r.h.s.~of (*) above \emph{fails} 
instead to belong to $L^3(\mathbb{R}^3)$. Since $u\in L^2_{\mathrm{ac}}(H_{\alpha,Y})$, 
then $u$ is orthogonal to all the eigenfunctions of $H_{\alpha,Y}$, that is, owing to 
Theorem \ref{thm:general_properties}(iv), $u$ is orthogonal to an 
(at most) $N$-dimensional subspace spanned by suitable linear combinations of 
$\mathcal{G}_{\ii\lambda_k}^{y_1},\dots,\mathcal{G}_{\ii\lambda_k}^{y_N}$ for 
$k\in\{1,\dots,N\}$, where $-\lambda^2_1,\dots,-\lambda^2_N$ are the eigenvalues of 
$H_{\alpha,Y}$. Because of our choice of $c$, in such an orthogonal complement there 
is surely $u$ which is \emph{not} orthogonal to the 
$\overline{\mathcal{G}_{\ii c}^{y_k}}$\,'s, namely,
\[
\int_{\mathbb{R}^3} \mathcal{G}_{\ii c}^{y_k}(y)\,u(y)\,\ud y\;\neq\; 
0\qquad \forall k\in\{1\ldots ,N\}\,.
\]
(In fact, such a $u$ can be also found in 
$C^\infty_0(\mathbb{R}^3)\cap L^2_{\mathrm{ac}}(H_{\alpha,Y})$: indeed, the point 
spectral subspace of $H_{\alpha,Y}$ is at most $N$-dimensional, whereas the set of 
$u$'s that satisfy the non-vanishing condition above is open in the topology of 
the space of test functions.)
For such $u$, because of the invertibility of the matrix 
$\Ga_{\alpha,Y}(\ii c)$, the expression
\[
\sum_{j,k=1}^N (\Gamma_{\alpha,Y}(\ii c)^{-1})_{jk}
\;\mathcal{G}_{\ii c}^{y_j}(x) 
\!\int_{\mathbb{R}^3} \mathcal{G}_{\ii c}^{y_k}(y)\,u(y)\,\ud y
\]
is a linear combination of the $\mathcal{G}_{\ii c}^{y_j}$'s with 
at least one non-zero coefficient, say, the one for $j=j_0$. Therefore, in a 
sufficiently small neighbourhood of $y_{j_0}$ (so small as not to contain any 
other of the $y_j$'s of $Y$, for $j\neq j_0$) the latter function must be of 
the form $c_{j_0}|x-y_{j_0}|^{-1}+R(x)$ for some constant $c_{j_0}\neq 0$ and 
some bounded (in fact, smooth) function $R(x)$. This would mean that in the 
considered neighbourhood of $y_{j_0}$ $R_{\alpha,Y}(-c^2) u - R_0(-c^2)u$ is 
\emph{not} a $L^3$-function, a contradiction.

\subsection{Unboundedness of $W^{+}_{\alpha,Y}$ in $L^1(\mathbb{R}^3)$} 
For this case the following preliminary observation is going to be useful.

\begin{rem}\label{re:hi}
Let $g\in C^{\infty}_0(\mathbb{R})$. Then
\begin{equation}\label{eq:Hilb-Fourier}
\frac{2}{\sqrt{2\pi}}\int_0^{+\infty}\!\! e^{-\ii\lambda \r}\,\widehat{g}(-\lambda)\,\ud\lambda\;=\;g(\rho)-\ii (\Hg g)(\r)\,,
\end{equation}
where $g\mapsto \cH g$ denotes the Hilbert transform, defined as
\begin{equation}\label{eq:def_Hilbert}
(\Hg g)(\r)\;:=\;\frac{1}{\pi}\,\mathrm{P.V.}\!\int_{-\infty}^{+\infty}\frac{g(\tau)}{\r-\tau}\,\ud\tau.
\end{equation}
Indeed, following from the fact \cite[Eq.~(5.1.13)]{Grafakos_ClassialFourier}
that the Hilbert transform is the Fourier multiplier
\[
\widehat{(\Hg g)}(\lambda)\;=\;-\ii\,\mathrm{sgn}(\lambda)\,\widehat{g}(\lambda)\,,
\]
one has
\begin{equation*}
 \begin{split}
  g(\rho)-\ii (\Hg g)(\r)\;&=\;\frac{1}{\sqrt{2\pi}}\int_{-\infty}^{+\infty}e^{\ii\r\lambda}(1-\mathrm{sgn}(\lambda))\,\widehat{g}(\lambda)\,\ud\lambda\\
&=\;\frac{2}{\sqrt{2\pi}}\int_{-\infty}^0e^{\ii\r\lambda}\,\widehat{g}(\lambda)\,\ud\lambda \;=\;\frac{2}{\sqrt{2\pi}}\int_0^{+\infty}\!\! e^{-\ii\lambda \r}\,\widehat{g}(-\lambda)\,\ud\lambda\,.
 \end{split}
\end{equation*}
\end{rem}

Let us now prove the fact that the wave operator $W_{\alpha,Y}^+$ is unbounded in $L^1(\mathbb{R}^3)$.
We may assume without loss of generality to take the set $Y=\{y_1,\dots,y_N\}$ of interaction centres so that $y_1=0$.

Let $u\in C^{\infty}_0(\mathbb{R}^3)$ be rotationally invariant, and we write $u(x)=f(|x|)$ for some $f:[0,+\infty)\rightarrow\mathbb{C}$ which is smooth and compactly supported. We extend $f$ to an even function on the whole $\mathbb{R}$. By construction, $f(r)=M_u(r)$, the spherical mean of $u$.

Our starting point is the stationary representation \eqref{eq:stationary_rep_w} for $W^{+}_{\a, Y} u$, that is,
\begin{equation}\label{eq:WOK}
W^{+}_{\a, Y} u \;=\; u  
+\sum_{j,k=1}^N T_{y_j} \Omega_{jk}T_{-y_k}u\,,
\end{equation}
and for each $j,k\in\{1,\ldots,N\}$ we set $K_{jk}u:=T_{y_j} \Omega_{jk}T_{-y_k}u$. Explicitly,
\begin{equation}\label{Kiju}
\begin{split}
 &(K_{jk}u)(x) \\
&\quad=\;\frac{1}{\ii\pi}\int_{\mathbb{R}^3}\!\ud y\,
u(y)\!\int_0^{+\infty}\!\!\ud \lambda\,F_{jk}(\lambda)\,
\frac{e^{-\ii\lambda|x-y_j|}}{4\pi|x-y_j|}\,
\frac{e^{\ii\lambda|y-y_k|}-e^{-\ii\lambda|y-y_k|}}{4\pi|y-y_k|} \,,
\end{split}
\end{equation}
where we used \eqref{eq:def_of_Omega_jk} and \eqref{eq:def_Fjk}.

We now proceed by re-scaling $u$ and $f$ as
\begin{equation}
 u_\ep(x)\;:=\; \ep^{-3}u(\ep^{-1}x)\,,\qquad f_\ep(r)\;:=\; \ep^{-3}f(\ep^{-1}r)\,,\qquad \varepsilon>0\,,
\end{equation}
which makes the norms
\begin{equation}\label{eq:norm-same-scaling}
 4\pi\|r^2 f_{\ep}\|_{L^1(0,+\infty)}\;=\;\Vert u_{\ep}\Vert_{L^1(\mathbb{R}^3)}\;=\;\Vert u\Vert_{L^1(\mathbb{R}^3)}\;=\;4\pi\|r^2 f\|_{L^1(0,+\infty)}
\end{equation}
$\ep$-independent. This re-scaling is devised so as to make all interaction centres but $y_1$ ineffective, because $u_\varepsilon$ is only bumped around the origin, and then to reduce the question to the unboundedness of the wave operator relative to a single-centre point interaction Hamiltonian, for which the answer will then come by direct inspection.

From \eqref{Kiju} and \eqref{eq:norm-same-scaling},
\begin{equation}\label{Kiju-eps}
\begin{split}
 &(K_{jk}u_\varepsilon)(x) \\
&\!\!\!\!=\;\frac{1}{\;\ii\pi\varepsilon^2}\int_{\mathbb{R}^3}\!\ud y\,u(y)\!\int_0^{+\infty}\!\!\ud \lambda\,F_{jk}({\textstyle\frac{\lambda}{\varepsilon}})\,
\frac{e^{-\ii\frac{\lambda}{\varepsilon}|x-y_j|}}{4\pi|x-y_j|}\,
\frac{e^{\ii\lambda|y-\frac{y_k}{\varepsilon}|}-e^{-\ii\lambda|y-\frac{y_k}{\varepsilon}|}}{4\pi|y- \frac{y_k}{\varepsilon}|} \,,
\end{split}
\end{equation}
having made the changes of variables $y \to \ep y$ and $\lambda \to \ep^{-1}\lambda$ in the integrations. If we now define, for arbitrary $v\in C^\infty_0(\mathbb{R}^3)$,
\begin{equation}\label{eq:de_k_scaled}
\begin{split}
 &(K_{jk}^{(\varepsilon)}v)(x) \\
&=\;\frac{1}{\ii\pi}\int_{\mathbb{R}^3}\!\ud y\,v(y)\!\int_0^{+\infty}\!\!\ud \lambda\,F_{jk}({\textstyle\frac{\lambda}{\varepsilon}})\,
\frac{e^{-\ii\lambda|x-\frac{y_j}{\varepsilon}|}}{4\pi|x-\frac{y_j}{\varepsilon}|}\,
\frac{e^{\ii\lambda|y-\frac{y_k}{\varepsilon}|}-e^{-\ii\lambda|y-\frac{y_k}{\varepsilon}|}}{4\pi|y-\frac{y_k}{\varepsilon}|} \,,
\end{split}
\end{equation}
then for the considered $u$ and its re-scaled $u_\varepsilon$ we have
\begin{equation}\label{eq:equiv_multi_k}
\Big\Vert\sum_{j,k=1}^NK_{jk}u_{\ep}\,\Big\Vert_{L^1(\mathbb{R}^3)}\;=\;\Big\Vert\sum_{j,k=1}^NK^{(\ep)}_{jk}u\,\Big\Vert_{L^1(\mathbb{R}^3)}\,,
\end{equation}
which follows by making the change of variable $x\mapsto \varepsilon x$ in the integration on the l.h.s.

We now want to study the contribution of each term $K_{jk}^{(\ep)}u$ as $\ep\downarrow 0$. We shall establish the following limits
\begin{equation}\label{eq:2limits}
 \begin{split}
  \lim_{\varepsilon\downarrow 0}\,(K_{11}^{(\ep)}u)(x)\;&=\;-\sqrt{\frac{2}{\pi}\,}\int_0^{+\infty}\frac{e^{-\ii\lambda|x|}}{|x|}\,(\widehat{rf})(-\lambda)\,\ud\lambda\,, \\
  \lim_{\varepsilon\downarrow 0}\,(K_{jk}^{(\ep)}u)(x)\;&=\;0\,,\qquad\quad (j,k)\neq(1,1)\,,
 \end{split}
\end{equation}
pointwise for a.e.~$x\in\mathbb{R}^3$. 

To this aim, we first find the bound
\begin{equation}\label{eq:estimate_Gl_G-l}
\int_{\mathbb{R}^3}\frac{\,e^{\ii\lambda|y-\frac{y_k}{\varepsilon}|}-e^{-\ii\lambda|y-\frac{y_k}{\varepsilon}|}}{4\pi|y-\frac{y_k}{\varepsilon}|} \, u(y)\,\ud y \;\leqslant_{|\;\cdot\;|} C_u \, \la\lambda\ra^{-2} \!\int_{\supp{u}}\frac{\ud y}{|y-\ep^{-1}y_k|}
\end{equation}
for some constant $C_u>0$ depending on $u$, but not on $\varepsilon$. \eqref{eq:estimate_Gl_G-l} is obvious for small $\lambda$'s, since $u$ is compactly supported, whereas for large $\lambda$'s we apply the distributional identity 
\begin{equation*}
(-\lap_y- \lambda^2)
\Big(
\frac{e^{\pm \ii\lambda|y-\frac{y_k}{\varepsilon}|}}{\,4\pi|y-\frac{y_k}{\varepsilon}|\,}\Big)
\;=\;\delta(y-{\textstyle \frac{y_k}{\varepsilon}})\,,
\end{equation*}
and integrating by parts we find 
\begin{equation*}
 \begin{split}
  \int_{\mathbb{R}^3}&\frac{\,e^{\ii\lambda|y-\frac{y_k}{\varepsilon}|}-e^{-\ii\lambda|y-\frac{y_k}{\varepsilon}|}}{4\pi|y-\frac{y_k}{\varepsilon}|} \, u(y)\,\ud y \\
 &=\; \lambda^{-2}
\int_{\mathbb{R}^3}\frac{\,e^{\ii\lambda|y-\frac{y_k}{\varepsilon}|}-e^{-\ii\lambda|y-\frac{y_k}{\varepsilon}|}}{4\pi|y-\frac{y_k}{\varepsilon}|} \,\,(-\lap) u(y) \ud y \\
&\leqslant_{|\;\cdot\;|} C_u \, \la\lambda\ra^{-2} \!\int_{\supp{u}}\frac{\ud y}{|y-\ep^{-1}y_k|}\,,
 \end{split}
\end{equation*}
thus, \eqref{eq:estimate_Gl_G-l} is proved.

Next, in order to prove the first of the limits \eqref{eq:2limits} by taking $\varepsilon\downarrow 0$ in \eqref{eq:de_k_scaled}, we use the asymptotics \eqref{eq:conv_Fjk_F}, namely,
\[
 \lim_{\ep\downarrow 0}\,F_{11}(\ep^{-1}\lambda)\;=\;-4\pi\ii\,,
\]
and we also recognise that the asymptotics as $\varepsilon\downarrow 0$ of the $y$-integration of \eqref{eq:de_k_scaled} is precisely the quantity 
\[
\int_{\mathbb{R}^3}\frac{\,e^{\ii\lambda|y|}-e^{-\ii\lambda|y|}}{4\pi|y|} \, u(y)\,\ud y\;=\;\sqrt{2\pi}\,(\widehat{rM_u})(-\lambda)\;=\;\sqrt{2\pi}\,(\widehat{r f})(-\lambda)
\]
discussed in \eqref{eq:rMu}. The limit $\varepsilon\downarrow 0$ can be exchanged with the 
integrations in $\lambda$ and in $y$ by dominated convergence, because 
$F_{11}(\frac{\lambda}{\varepsilon})$  is uniformly bounded (see Lemma \ref{lm:3-2}(i))  
and \eqref{eq:estimate_Gl_G-l} provides 
a majorant that is integrable in $\lambda$. Thus,
\[
\begin{split}
 \lim_{\varepsilon\downarrow 0}\,(K_{11}^{(\ep)}u)(x)\;&=\;\frac{1}{\ii\pi}(-4\pi\ii)\int_0^{+\infty}\frac{e^{-\ii\lambda|x|}}{4\pi|x|}\,\sqrt{2\pi}\,(\widehat{r f})(-\lambda)\,\ud\lambda \\
 &=\; 
 -\sqrt{\frac{2}{\pi}\,}\int_0^{+\infty}\frac{e^{-\ii\lambda|x|}}{|x|}\,(\widehat{rf})(-\lambda)\,\ud\lambda\,, \quad x \not=0\,,
 \end{split}
\]
and the first limit of \eqref{eq:2limits} is proved.

Concerning now \eqref{eq:2limits} when $(j,k)\neq (1,1)$, from our estimate \eqref{eq:estimate_Gl_G-l} we deduce
\begin{equation}\label{eq:full_bound_delta}
\begin{split}
\!\!\!|(K^{(\ep)}_{jk}u)(x)|\;&\leqslant\; \frac{C_u}{|x-\frac{y_j}{\varepsilon}|}\Big(\int_0^{+\infty}|F_{jk}({\textstyle\frac{\lambda}{\varepsilon}})|\,\langle\lambda\rangle^{-2}\,\ud\lambda\Big)\Big(\int_{\supp{u}}\frac{\ud y}{|y-\frac{y_k}{\varepsilon}|}\Big)\\
&\leqslant\; C'_u\,\Vert F_{jk}\Vert_{L^{\infty}(0,\infty)}\,\frac{1}{|x-\frac{y_j}{\varepsilon}|}\,\int_{\supp{u}}\frac{\ud y}{|y-\frac{y_k}{\varepsilon}|}
\end{split}
\end{equation}
for some new constant $C'_u>0$. Since at least one among $y_j$ and $y_k$ does not coincide with the origin, and since $u$ is compactly supported, we conclude at once that
\begin{equation*}
\lim_{\ep\downarrow 0}(K_{jk}^{(\ep)}u)(x)\;=\;0\, , \quad x\not=0 \ \ \mbox{if $j=0$.} 
\end{equation*}

The proof of \eqref{eq:2limits} is thus completed, and in turn \eqref{eq:2limits} implies
\begin{equation}\label{eq:puntual3}
\lim_{\ep\downarrow 0} 
\sum_{j,k=1}^N (K^{(\ep)}_{jk}u)(x)\;=\;-\frac{2}{\sqrt{2\pi}}\int_0^{+\infty}\frac{e^{-\ii\lambda|x|}}{|x|}\,(\widehat{rf})(-\lambda)\,\ud\lambda
\end{equation}
pointwise for a.e.~$x\in\mathbb{R}^3$. 

This latter fact allows us to take the limit $\varepsilon\downarrow 0$ in 
the r.h.s.~of \eqref{eq:equiv_multi_k}, provided that the $L^1$-norm is taken on 
compacts of $\mathbb{R}^3$. Indeed, for fixed $R>0$ and any sufficiently small 
$\varepsilon>0$ 
such that $|x-\frac{y_j}{\varepsilon}|\geqslant |x|$ for any 
$x \in \{x | |x| \leqslant R\} \cup \supp{u}$ and 
$j=1, \dots, N$, the estimate \eqref{eq:full_bound_delta} implies 
($\mathbf{1}_{\!R}\equiv$ the characteristic function of the ball $|x|\leqslant R$)
\[
\mathbf{1}_{\!R}(x)\!\sum_{j,k=1}^N|(K^{(\ep)}_{jk}u)(x)|\;
\leqslant\; N^2\,\frac{C_{u,R}}{|x|} \!\! \int_{\supp{u}}
\frac{\ud y}{|y|}\;\leqslant\; N^2\,\frac{C'_{u,R}}{|x|}
\]
for suitable constants $C_{u,R},C'_{u,R}>0$, which gives a majorant in $L^1(\mathbb{R}^3)$. Then, by \eqref{eq:puntual3} and dominated convergence,
\begin{equation}\label{eq:conv_uni_int}
\begin{split}
 \lim_{\ep\downarrow 0}\int_{|x|\leqslant R}
&\Big|\!\sum_{j,k=1}^N (K^{(\ep)}_{jk}u)(x)\Big|\,\ud x \\
&\qquad=\;\frac{2}{\sqrt{2\pi}}\int_{|x|\leqslant R}\ud x\;\Big|\!\int_0^{+\infty}\!\!\ud\lambda\,\frac{e^{-\ii\lambda|x|}}{|x|}\,(\widehat{rf})(-\lambda)\,\Big|\\
&\qquad=\;\sqrt{32\pi}\int_0^R\!\ud \rho\;\Big|\!\int_0^{\infty}\r\, e^{-\ii\lambda\r}\,(\widehat{rf})(-\lambda)\,\ud\lambda\,\Big|\,.
\end{split}
\end{equation}
An integration by parts and formula \eqref{eq:Hilb-Fourier} in Remark \ref{re:hi} yield
\begin{equation}\label{eq:furtherIntPart}
\begin{split}
 \sqrt{32\pi}\int_0^{\infty}\r\, e^{-\ii\lambda\r}\,(\widehat{rf})(-\lambda)\,\ud\lambda\;&=\;\sqrt{32\pi}\int_0^{+\infty}\!e^{-\ii\lambda\r}\,(\widehat{r^2f})(-\lambda)\,\ud\lambda\\
&=\;4\pi\,\big((r^2 f)(\rho)-\ii(\cH r^2 f)(\rho)\big).
\end{split}
\end{equation}
In the integration by parts the boundary term does not appear because $r\mapsto r f(r)$ is an odd function and $(\widehat{r f})(0)=0$. The conclusion from \eqref{eq:conv_uni_int} and \eqref{eq:furtherIntPart} is therefore
\begin{equation}\label{eq:final_conv}
 \lim_{\varepsilon\downarrow 0}\,\Big\|\,\mathbf{1}_{\!R}\!\!\sum_{j,k=1}^NK_{j,k}^{(\ep)}u\,\Big\|_{L^1(\mathbb{R}^3)}\;=\;4\pi\int_0^R \big|(\mathbbm{1}-\ii\Hg)(r^2f)(\r)\big|\,\ud\r\,.
\end{equation}

The proof of the $L^1$-unboundedness of $W_{\alpha,Y}^+$ is completed as follows. Suppose for contradiction that $W_{\alpha,Y}^+$ is instead  $L^1$-bounded. Then, for arbitrary $R>0$,

\begin{equation*}
\begin{split}
4\pi\int_0^R \big|(\mathbbm{1}-\ii\Hg)(r^2f)(\r)\big|\,\ud\r\;&=\;\lim_{\varepsilon\downarrow 0}\,\Big\|\,\mathbf{1}_{\!R}\!\!\sum_{j,k=1}^NK_{j,k}^{(\ep)}u\,\Big\|_{L^1(\mathbb{R}^3)}\\
&\leqslant\;\liminf_{\ep\downarrow 0}\,\Big\|\sum_{j,k=1}^NK_{j,k}^{(\ep)}u\,\Big\|_{L^1(\mathbb{R}^3)} \\
&=\; \liminf_{\ep\downarrow 0}\,\Big\|\sum_{j,k=1}^NK_{j,k}u_\varepsilon\,\Big\|_{L^1(\mathbb{R}^3)}\\
&=\;\liminf_{\ep\downarrow 0}\Vert(W_{\alpha,Y}^+-\mathbbm{1})u_{\ep}\Vert_{L^1(\mathbb{R}^3)}\\
&\leqslant\; (1+\Vert W_{\alpha,Y}^+\Vert_{\mathcal{B}(L^1(\mathbb{R}^3))})\Vert u_{\ep}\Vert_{L^1(\mathbb{R}^3)}\\
&\leqslant \;(1+\Vert W_{\alpha,Y}^+\Vert_{\mathcal{B}(L^1(\mathbb{R}^3))})\Vert r^2f\Vert_{L^1(0,\infty)}\,,
\end{split}
\end{equation*}
where we applied \eqref{eq:final_conv} in the first step, \eqref{eq:equiv_multi_k} in the third step, \eqref{eq:WOK} in the fourth step, the assumption of $L^1$-boundedness in the fifth step, and the scale invariance \eqref{eq:norm-same-scaling} in the last step. Moreover, due to the arbitrariness of $R$, the estimate above also implies
\begin{equation*}\tag{*}
 4\pi \|(\mathbbm{1}-\ii\Hg)(r^2f)\|_{L^1(0,\infty)}\;\leqslant \;(1+\Vert W_{\alpha,Y}\Vert_{\mathcal{B}(L^1(\mathbb{R}^3))})\Vert r^2f\Vert_{L^1(0,\infty)}\,.
\end{equation*}
However, the inequality (*) can be surely violated. Indeed it is well-known that the Hilbert transform on $\mathbb{R}$ maps \emph{even} functions into \emph{odd} functions, but fails to map even (and compactly supported) $L^1$-functions into $L^1$-functions, as one may see with (a suitable mollification, so as to make it $C^\infty_0$ and even, of) the function $f_0(r)=(r^2+1)^{-1}$, the Hilbert transform of which is $(\cH f_0)(r)=r(r^2+1)^{-1}$. Therefore (*) is a contradiction. The conclusion is that $W_{\alpha,Y}^+$ is necessarily \emph{unbounded} on $L^1(\mathbb{R}^3)$.

\section{$L^p$-convergence of wave operators}\label{sec:convergence_of_wave}

In this concluding Section we establish a result of $L^p$-convergence of wave operators in the limit when a regular Schr\"{o}dinger Hamiltonian converges to a singular point interaction Hamiltonian. This is part of the general picture outlined in Remark \ref{re:consistent} concerning the connection between two completely analogous results, on the one hand our main result (Theorem \ref{ref:main_thm}) of $L^p$-boundedness for $p\in(1,3)$ and $L^p$-unboundedness for $p\in\{1\}\cup[3,\infty)$ of the wave operators relative to the point interaction Hamiltonian $H_{\alpha,Y}$, and on the other hand the analogous results available in the previous literature, precisely in the same regimes of $p$, for wave operators relative to Schr\"{o}dinger Hamiltonians of the form $-\Delta+V$.

For concreteness we restrict our attention to the case $N=1$ and $\alpha=0$, thus taking without loss of generality $Y=\{0\}$.

Besides the corresponding point interaction Hamiltonian $H_{\alpha,Y}$ and wave operators $W^{\pm}_{\a, Y}$, let us consider the Schr\"{o}dinger operator $H=-\Delta+V$ on $L^2(\mathbb{R}^3)$, where $V$ is a real measurable potential such that $|V(x)|\leqslant C \ax^{-\delta}$ for some $\delta>5/2$. It is well known 
\cite{Kuroda1978_intro_scatt_theory}  that $H$ has a unique self-adjoint 
realisation on $L^2(\mathbb{R}^3)$ and that the wave operators
\begin{equation}\label{eq:WopSchr}
W^{\pm}\;:=\;s\textrm{-}\!\!\!\!\lim_{t\to\pm \infty} e^{\ii tH}e^{-\ii tH_0}
\end{equation}
relative to the pair $(H,H_0)$ exist and are complete in $L^2(\mathbb{R}^3)$; 
$W^{\pm}$ extend to bounded operators on $L^p(\mathbb{R}^3)$ in the following regimes: for all 
$p\in[1,+\infty]$ if zero is neither a resonance nor eigenvalue of $H$ 
\cite{Beceanu-Schlag-2016}, and only for $p\in(1,3)$ if zero is a resonance \cite{Yajima-DocMath2016} (see Proposition 
\ref{prop:limit_of_wave_ops} below). 

Parallel to that, and with the same $V$, we consider now the re-scaled version $H^{(\varepsilon)}$ of $H$ obtained by `shrinking' the potential $V$ at a scale $\varepsilon^{-1}$, more precisely, the self-adjoint operator
\begin{equation}\label{eq:HepsShrinked}
 H^{(\varepsilon)}\;:=\;-\Delta+\frac{1}{\ep^2}V\left(\frac{x}{\ep}\right)\,,\qquad\varepsilon>0\,,
\end{equation}
as well as the wave operators relative to the pair $(H^{(\varepsilon)},H_0)$, defined in analogy to \eqref{eq:WopSchr} as
\begin{equation}\label{eq:WopSchr-eps}
W^{\pm}_\varepsilon\;:=\;s\textrm{-}\!\!\!\!\lim_{t\to\pm \infty} e^{\ii tH^{(\varepsilon)}}e^{-\ii tH_0}\,.
\end{equation}

The choice \eqref{eq:HepsShrinked} for the scaling is driven by the fact 
\cite[Theorem I.1.2.5]{albeverio-solvable} that under suitable 
spectral properties of 
$H$ one has $H^{(\varepsilon)}\to H_{\alpha,Y}|_{\alpha=0,Y=\{0\}}$ as 
$\varepsilon\downarrow 0$ in the norm resolvent sense of operators on 
$L^2(\mathbb{R}^3)$, and this in turn motivates us to investigate the relation 
between $W^{\pm}_\varepsilon$ and $W^{\pm}_{\a, Y}$ when $\varepsilon\downarrow 0$, 
as bounded operators on $L^p(\mathbb{R}^3)$ for $p\in(1,3)$. 
Our result is the following.

\begin{prop}\label{prop:limit_of_wave_ops} Suppose that $V$ is a real measurable 
potential such that $|V(x)|\leqslant C \ax^{-\delta}$ for some $\delta>7$. 
Then, for any $p\in(1,3)$ the wave operators 
$W^\pm_\varepsilon$ extend to bounded operators on $L^p(\mathbb{R}^3)$. 
If zero is a resonance but not an eigenvalue for the self-adjoint 
operator $H=-\Delta+V$ on $L^2(\mathbb{R}^3)$, then  in the weak 
topology of $L^p(\mathbb{R}^3)$ with $p\in(1,3)$, and hence also in 
the strong topology of $L^2(\mathbb{R}^3)$,
\begin{equation}\label{eq:Wop_convergence}
\lim_{\ep \downarrow 0}W^{\pm}_{\ep}u \;=\; W^{\pm}_{\alpha,Y}u\,, 
\qquad u \in L^p(\mathbb{R}^3)\,.
\end{equation}
\end{prop}

\begin{proof} The statement on the $L^p$-boundedness of $W^\pm_\varepsilon$ follows directly from \cite{Yajima-DocMath2016}. Concerning the limit \eqref{eq:Wop_convergence}, we shall prove it for $W^+_\varepsilon$, the argument for $W^-_\varepsilon$ being completely analogous. 

Let us consider the scaling operator $u\mapsto U_\varepsilon u$ defined by
\begin{equation}
 (U_\varepsilon u)(x)\;:=\;\frac{1}{\;\varepsilon^{3/2}}\,u\Big(\frac{x}{\varepsilon}\Big)\,,\qquad \varepsilon>0\,,\,x\in\mathbb{R}^3\,.
\end{equation}
For any $\varepsilon>0$ and  $p\in[1,+\infty]$ the operator $U_\varepsilon$ is a bounded bijection on $L^p(\mathbb{R}^3)$ with norm
\begin{equation}\label{eq:Unorm}
 \|U_\varepsilon\|_{\mathcal{B}(L^p(\mathbb{R}^3))}\;=\;\varepsilon^{\,3(\frac{1}{p}-\frac{1}{2})}
\end{equation}
and inverse
\begin{equation}\label{eq:Uinverse}
 (U_{\varepsilon})^{-1}\;=\;U_{\varepsilon^{-1}}\,.
\end{equation}
In particular $U_\varepsilon$ is unitary on $L^2(\mathbb{R}^3)$, and it induces the unitary equivalence
\begin{equation} 
H^{(\varepsilon)} \;=\; U_{\ep} (\ep^{-2}H) \,U_{\ep}^\ast. 
\end{equation}
As a consequence, $W_\varepsilon^+$ and $W^+$ are unitarily equivalent too as operators on $L^2(\mathbb{R}^3)$, for
\begin{equation}\label{eq:unitary_equiv} 
\begin{split}
 W_{\ep}^+\;&=\;s\textrm{-}\!\!\!\!\lim_{t\to +\infty} e^{it H^{(\varepsilon)}}e^{-itH_0} \\
 &=\;
U_\ep \,s\textrm{-}\!\!\!\!\lim_{t\to + \infty} e^{it\ep^{-2}H}
e^{-it\ep^{-2}H_0} \,U_{\ep}^\ast 
\;=\;U_{\ep}W^+U_{\ep}^{\ast} \, .
\end{split}
\end{equation}
Moreover,
\begin{equation}\label{eq:Wepsnorm}
\|W_\varepsilon^+\| _{\mathcal{B}(L^p(\mathbb{R}^3))}\;=\;\|W^+\| _{\mathcal{B}(L^p(\mathbb{R}^3))}\;<\;+\infty
\end{equation}
for any $p\in(1,3)$, as follows by combining \eqref{eq:Unorm}, \eqref{eq:Uinverse}, and \eqref{eq:unitary_equiv}.

For the proof of \eqref{eq:Wop_convergence} it suffices to show that, when $\alpha=0$ and $Y=\{0\}$,
\begin{equation}\label{eq:limit_to_prove1}
 \lim_{\varepsilon\downarrow 0}\int_{\mathbb{R}^3} 
\overline{(W_\varepsilon^+ u)(x)}\,v(x)\,\ud x\;
=\int_{\mathbb{R}^3} \overline{(W_{\alpha,Y} u)(x)}\,v(x)
\end{equation}
for any $u$ and $v$ in 
\begin{equation}\label{eq:choice-of-the-dense}
 \mathcal{D}\;:=\;\{u\in\mathcal{S}(\mathbb{R}^3)
\,|\,\widehat{u}\in C^\infty_0(\mathbb{R}^3)\}\, 
\end{equation}
which is dense in $L^p(\mathbb{R}^3)$ for any $1<p<\infty$. 
Indeed by means of a straightforward density argument, applicable because of 
the uniform norm-boundedness \eqref{eq:Wepsnorm}, the result 
\eqref{eq:limit_to_prove1} can then be lifted to any 
 $u\in L^p(\mathbb{R}^d)$ and $v\in L^{p'}(\mathbb{R}^d)$, 
whence the conclusion.
Moreover, with the choice \eqref{eq:choice-of-the-dense} we can  equivalently
re-write \eqref{eq:limit_to_prove1} in Hilbert scalar product notation as  
\begin{equation}\label{eq:limit_to_prove2}
 \lim_{\varepsilon\downarrow 0}\;\langle W_\varepsilon^+ u,v\rangle\;
=\; \langle W_{\alpha,Y} u,v\rangle\,.
\end{equation}

Aimed at establishing \eqref{eq:limit_to_prove2}, let us fix $u,v\in\mathcal{D}$. Then there is $R>0$ such that $\widehat{u}(\xi)=0$ for $|\xi|>R$, and also
\begin{equation}\label{eq:Uutransf}
 (\widehat{U_\varepsilon^* u)}(\xi)\;=\;\frac{1}{\;\varepsilon^{3/2}}\,\widehat{u}\Big(\frac{\xi}{\varepsilon}\Big)\,,\qquad (\widehat{U_\varepsilon^* u)}(\xi)\;=\;0\quad\textrm{for }|\xi|>R\varepsilon\,.
\end{equation}
We shall make crucial use of the well-known fact from the stationary scattering theory \cite{Kuroda1978_intro_scatt_theory} that
\begin{equation}\label{eq:W_scatt_theory}
 W^+\;=\;\mathbbm{1}-\frac{1}{\ii\pi}\int_0^{+\infty}\!\! G_0(-\lambda)\,V\,(\mathbbm{1}+G_0(-\lambda)V)^{-1}\,(G_0(\lambda)-G_0(\lambda))\,\lambda\,\ud\lambda\,,
\end{equation}
where
\begin{equation}
 G_0(\pm\lambda)\;:=\;\lim_{\eta\downarrow 0}\,(H_0-(\lambda^2\pm\ii\eta)\mathbbm{1})^{-1}\;=\;\lim_{\eta\downarrow 0}\, R_0(\lambda^2\pm\ii\eta)\,,\quad\lambda\geqslant 0\,.
\end{equation}
Then \eqref{eq:unitary_equiv} and \eqref{eq:W_scatt_theory}, together with $G_0(\pm\lambda)^*=G_0(\mp\lambda)$, yield
\begin{equation}\label{eq:uWv_1}
 \begin{split}
  &\langle W_\varepsilon^+u,v\rangle-\langle u,v\rangle \\
  &=\;\frac{1}{\ii\pi}\int_0^{+\infty}\!\!\big\langle (\mathbbm{1}+G_0(-\lambda)V)^{-1}\,(G_0(\lambda)-G_0(-\lambda))\,U_\varepsilon^*u\;,V\,G_0(\lambda)\,U_\varepsilon^* v\big\rangle\,\lambda\,\ud\lambda\,. 
 \end{split}
\end{equation}

In fact, the $\lambda$-integration in \eqref{eq:uWv_1} is only effective for $\lambda<R\varepsilon$. To see this, we compute the Fourier transform
\begin{equation}\label{eq:Ftrans}
 \begin{split}
  &\big((G_0(\lambda)-G_0(-\lambda))\,U_\varepsilon^*u\big)^{\widehat{\textrm{ }}}(\xi) \\
  &\qquad=\;\lim_{\eta\downarrow 0}\,\big((\xi^2-\lambda^2-\ii\eta)^{-1}-(\xi^2-\lambda^2+\ii\eta)\big)^{-1}\,(\widehat{U_\varepsilon^*u})(\xi) \\
  &\qquad=\;\lim_{\eta\downarrow 0}\,\frac{2\,\ii\,\eta}{\,(\xi^2-\lambda^2)^2+\eta^2\,}\,(\widehat{U_\varepsilon^*u})(\xi)
 \end{split}
\end{equation}
and we argue that the function in \eqref{eq:Ftrans} surely vanishes when $|\xi|>R\varepsilon$, owing to \eqref{eq:Uutransf}, and when in addition $\lambda>R\varepsilon$ such function also vanishes when $|\xi|\leqslant R\varepsilon$, because in this case $(\xi^2-\lambda^2)^2>0$ and the above limit in $\eta$ is zero. Thus,
\begin{equation}\label{eq:only_effective}
 (G_0(\lambda)-G_0(-\lambda))\,U_\varepsilon^*u\;\equiv\;0\qquad\textrm{when }\lambda>R\varepsilon\,.
\end{equation}

By exploiting the scaling in $\varepsilon$ in \eqref{eq:uWv_1} we obtain
\begin{equation}\label{eq:uWv_2}
 \begin{split}
  &\langle W_\varepsilon^+u,v\rangle-\langle u,v\rangle \\
  &=\;\frac{\;\varepsilon^2}{\ii\pi}\int_0^{+\infty}\!\!\big\langle (\mathbbm{1}+G_0(-\varepsilon\lambda)V)^{-1}(G_0(\varepsilon\lambda)-G_0(-\varepsilon\lambda))\,U_\varepsilon^*u\;,VG_0(\varepsilon\lambda)\,U_\varepsilon^* v\big\rangle\lambda\,\ud\lambda\,,
 \end{split}
\end{equation}
where it has to be remembered that, owing to \eqref{eq:only_effective}, the integration actually only takes place when $\lambda\in[0,R]$.

Next, in order to compute the limit $\varepsilon\downarrow 0$ in \eqref{eq:uWv_2}, we consider separately the behaviour of the operators
\[
 \varepsilon^{\frac{1}{2}} G_0(\pm\varepsilon\lambda)\,U_\varepsilon^*\qquad\textrm{and}\qquad\varepsilon(\mathbbm{1}+G_0(-\varepsilon\lambda)V)^{-1}\,.
\]
Indeed, we shall see that they do converge strongly in a suitable Banach space. 
A weak-type H\"older's inequality implies that
\[
\big( \varepsilon^{\frac{1}{2}} G_0(\pm\varepsilon\lambda)\,U_\varepsilon^*u\big)(x)\;=\;
\int_{\mathbb{R}^3}\frac{\,e^{\pm\ii\lambda|\varepsilon x-y|}\,}{\,4\pi|\varepsilon x-y|\,}\,u(y)\,\ud y 
\]
is bounded by a constant (by $(4\pi)^{-1}\||x|^{-1}\|_{L^{3,\infty}} \|u\|_{L^{\frac32,1}}$ in terms of  
Lorentz norms); therefore, 
uniformly for $\lambda\in[0,R]$ and $x$ in compact sets,
\[
\lim_{\varepsilon\downarrow 0}\,\big(\varepsilon^{\frac{1}{2}} G_0(\pm\varepsilon\lambda)\,U_\varepsilon^*u\big)(x)\; =\; \int_{\mathbb{R}^3}\frac{\,e^{\pm\ii\lambda|y|}\,}{\,4\pi|y|\,}\,u(y)\,\ud y\;=\;\langle\,\overline{\mathcal{G}_{\pm\lambda}},u\rangle\, . 
\]
As a consequence, we deduce that 
\begin{equation}\label{eq:op_bound_1}
 \lim_{\varepsilon\downarrow 0}\,\big\|\varepsilon^{\frac{1}{2}} G_0(\pm\varepsilon\lambda)\,U_\varepsilon^*u-\langle\,\overline{\mathcal{G}_{\pm\lambda}},u\rangle\mathbf{1}\big\|_{L^2_{-\beta}(\mathbb{R}^3)}=\;0
\end{equation}
for $\beta>\frac{3}{2}$,
where $L^2_{-\beta}(\mathbb{R}^3)\equiv L^2(\mathbb{R}^3,\langle x\rangle^{-2\beta}\ud x)$ and $\mathbf{1}$ denotes the function $\mathbf{1}(x)=1$ $\forall x\in\mathbb{R}^3$.
Moreover, owing to the spectral and decay assumptions on $V$, it is a standard fact \cite[Theorem 4.8]{Yajima-CMP2015_dispersive_est} that 
\begin{equation}\label{eq:op_bound_2}
 \lim_{\varepsilon\downarrow 0}\,\Big\|\,\varepsilon(\mathbbm{1}+G_0(-\varepsilon\lambda)V)^{-1}-\frac{4\pi\ii}{\,\lambda {a}^2}\,|\varphi\rangle\langle V \varphi|\,\Big\|_{\mathcal{B}(L^2_{-\beta}(\mathbb{R}^3))}=\;0
\end{equation}
for $\lambda>0$ and $\beta\in(\frac{3}{2},\delta-\frac{1}{2})$, where 
$\varphi$ is the so-called `resonance function' 
relative to $V$ (thus, a distributional solution to $H\varphi=(-\Delta+V)\varphi=0$),
uniquely identified by the conditions $\int_{\mathbb{R}^3}V|\varphi|^2\ud x=-1$ 
and $\int_{\mathbb{R}^3}\!V\varphi\,\ud x>0$, and where $ a:= \int_{\mathbb{R}^3}\!V\varphi\,\ud x $.

If we now and henceforth restrict $\beta$ to the regime $\beta\in(\frac{3}{2},\frac{\delta}{2})$, then \eqref{eq:op_bound_1} and \eqref{eq:op_bound_2} are still valid, and in addition the multiplication by $V$ is a $L^2_{-\beta}(\mathbb{R}^3)\to L^2_{\beta}(\mathbb{R}^3)$ continuous map. Thus, the $L^2$-scalar product appearing in the r.h.s.~of \eqref{eq:uWv_2} can be  also regarded as a $L^2_{-\beta}$-$L^2_\beta$ duality product.
Using this fact, and by means of \eqref{eq:op_bound_1} and \eqref{eq:op_bound_2}, which are applicable because the $\lambda$-integration in \eqref{eq:uWv_2} is actually only effective 
for $\lambda\in[0,R]$, we find
\[
 \begin{split} 
& \lim_{\varepsilon\downarrow 0}\,(\mbox{r.h.s. of \eqref{eq:uWv_2}}) \\
  &=\;\frac{1}{\ii\pi}\int_0^{+\infty}\!\!\Big\langle\,\frac{4\pi\ii\;\;}{\lambda {a}^2}\,|\varphi\rangle\langle V\varphi|\big(\langle\overline{\mathcal{G}_\lambda},u\rangle-\langle\overline{\mathcal{G}_{-\lambda}},u\rangle\big)\mathbf{1}\,,V\langle\overline{\mathcal{G}_{\lambda}},v\rangle\,\mathbf{1}\Big\rangle_{\!L^2_{-\beta},L^2_\beta}\,\lambda\,\ud\lambda \\
  &=\;-4\int_0^{+\infty}\!\!\ud\lambda\,\Big(\int_{\mathbb{R}^3}\!\ud y\;\overline{u(y)}\,\big(\mathcal{G}_{-\lambda}(y)-\mathcal{G}_{\lambda}(y)\big)\Big)\Big(\int_{\mathbb{R}^3}\!\ud x\,\mathcal{G}_{\lambda}(x)\,v(x)\Big)\,.
 \end{split} 
\]

Summarising, we have found
\begin{equation}
 \begin{split}
  &\lim_{\varepsilon\downarrow 0}\,\langle W_\varepsilon^+u,v\rangle\;=\;\langle u,v\rangle \\
  &\quad +4\!\int_0^{+\infty}\!\!\ud\lambda\,\Big(\int_{\mathbb{R}^3}\!\ud y\;\overline{u(y)}\,\big(\mathcal{G}_\lambda(y)-\mathcal{G}_{-\lambda}(y)\big)\Big)\Big(\int_{\mathbb{R}^3}\!\ud x\,\mathcal{G}_{\lambda}(x)\,v(x)\Big)\,.
 \end{split}
\end{equation}
Since the r.h.s.~above is precisely the quantity 
$\langle W_{\alpha, Y}u,v\rangle$ that we obtained in 
\eqref{eq:def_of_Omega_jk} 
in the special case $N=1$, $\alpha=0$, the limit 
$\langle W_\varepsilon^+u,v\rangle\to \langle W_{\alpha, Y}u,v\rangle$ of 
\eqref{eq:limit_to_prove2} is then established and, as already argued, 
this completes the proof.
\end{proof}



\def\cprime{$'$}

\end{document}